\documentclass[12pt]{article}
\pdfoutput=1
\usepackage{amsmath,amssymb,amsthm,bm,graphicx,float,array,multirow,multicol,rotfloat,caption,subcaption,hyperref,cleveref,enumerate,geometry,mathdots,adjustbox,booktabs,parskip,mathtools,tikz,pdflscape,csquotes,lscape,rotating,empheq,braket,tikz-cd}
\usepackage[para]{threeparttable}
\usepackage[all]{xy}
\usepackage[normalem]{ulem}
\usepackage[numbers,sort&compress]{natbib}
\numberwithin{equation}{section}
\numberwithin{figure}{section}
\allowdisplaybreaks
\makeatletter
\usetikzlibrary{arrows, positioning, decorations.pathmorphing, decorations.markings, decorations.pathreplacing, decorations.markings, matrix, patterns}
\hypersetup{colorlinks=true}
\hypersetup{linkcolor=black}
\hypersetup{citecolor=black}
\hypersetup{urlcolor=black}
\makeatletter

\theoremstyle{plain}
\newtheorem*{thm*}{Theorem}
\newtheorem{thm}{Theorem}[section]

\newtheorem{lem}[thm]{Lemma}
\newtheorem{prop}[thm]{Proposition}
\newtheorem{conj}[thm]{Conjecture}

\theoremstyle{definition}
\newtheorem{defn}[thm]{Definition}
\newtheorem*{defn*}{Definition}

\crefname{lemma}{lemma}{lemmas}
\Crefname{lemma}{Lemma}{Lemmas}
\crefname{thm}{theorem}{theorems}
\Crefname{thm}{Theorem}{Theorems}
\crefname{defn}{definition}{definitions}
\Crefname{defn}{Definition}{Definitions}

\DeclarePairedDelimiterX{\abs}[1]{\lvert}{\rvert}{\ifblank{#1}{{}\cdot{}}{#1}}
\makeatother

\newtheorem*{thm:main}{Theorem \ref{thm:main}}
\newtheorem*{thm:prop}{Proposition \ref{thm:prop}}

\begin{document}

\begin{titlepage}
% Report number
\vspace*{-3cm} 
\begin{center}
\vspace{2.2cm}
{\LARGE\bfseries Large $N$ von Neumann algebras and the renormalization of Newton's constant}\\
\vspace{1cm}
{\large
Elliott Gesteau\\}
\vspace{.6cm}
{ Division of Physics, Mathematics, and Astronomy, California Institute of Technology}\par\vspace{-.3cm}
\vspace{.4cm}

\scalebox{.95}{\tt  egesteau@caltech.edu}\par
\vspace{1cm}
{\bf{Abstract}}\\
\end{center}
I derive a family of Ryu--Takayanagi formulae that are valid in the large $N$ limit of holographic quantum error-correcting codes, and parameterized by a choice of UV cutoff in the bulk. The bulk entropy terms are matched with a family of von Neumann factors nested inside the large $N$ von Neumann algebra describing the bulk effective field theory. These factors are mapped onto one another by a family of conditional expectations, which are interpreted as a renormalization group flow for the code subspace. Under this flow, I show that the renormalizations of the area term and the bulk entropy term exactly compensate each other. This result provides a concrete realization of the ER=EPR paradigm, as well as an explicit proof of a conjecture due to Susskind and Uglum.
\\
\vfill 
%{ \today\   at \ \currenttime\par}
\end{titlepage}

\tableofcontents
\newpage
\section{Introduction}

The quantum extremal surface (QES) formula \cite{Ryu:2006bv,Hubeny:2007xt,Lewkowycz:2013nqa,Engelhardt:2014gca,Harlow:2016vwg,Dong_2018} is one of the most important results in holography. It is a powerful probe of the emergent geometry of spacetime, as it relates the area of a surface in the bulk to the entanglement entropy of the boundary dual region. More precisely, if $\rho$ is the state of the boundary theory, it holds that
\begin{align}
    S(\rho)=\frac{A(\Sigma)}{4G_N}+S_{bulk}(\rho),
\end{align}
where $S(\rho)$ is the UV-complete entanglement entropy of the boundary state, $A(\Sigma)$ is the area of the \textit{quantum extremal surface} associated to $\rho$, $G_N$ is Newton's gravitation constant, and $S_{bulk}(\rho)$ is the entanglement entropy of $\rho$ \textit{in the bulk effective field theory}.

As crucial as this result is, the precise definitions of the terms in the QES formula remain elusive, and some paradoxes seem to arise. Namely:

\begin{itemize}
    \item $S_{bulk}(\rho)$ is divergent in the bulk effective field theory. Indeed, it is a general property of quantum field theory that the entanglement entropy of bulk subregions diverges, and one needs to regulate it with a UV cutoff. A more abstract formulation of this fact is that the von Neumann algebra of the bulk effective field theory has type $III_1$, or $II_\infty$ if one includes perturbative corrections and quantizes the ADM mass \cite{Leutheusser:2021frk,Leutheusser:2021qhd,Witten:2021unn}.\footnote{Technically, the same kind of issue arises for the boundary entropy associated to a nontrivial subregion. However, when considering two boundary theories in the thermofield double state, or in a large class of more general entangled states, this kind of divergence doesn't occur, and this is less of a serious problem in one wants to get to a conceptual understanding of holography. As a result, I will mostly consider boundary regions that are dual to one side of a black hole, and ignore this extra complication here.} 
    \item The smooth spacetime description of the bulk is only valid in the $G_N=0$ limit. In this limit, the area term also blows up and it is not clear what approximations need to be made in the bulk to consider $G_N$ small but nonzero.
    \item The bulk effective field theory seems to perform a calculation at $G_N=0$, or at least perturbatively in $G_N$. However, on the boundary, the calculation of $S(\rho)$ corresponds to a calculation in a conformal field theory at \textit{large but finite} $N$, so at \textit{small but finite} $G_N$, as $G_N\sim\frac{1}{N^2}$ from the holographic dictionary. At $N=\infty$, the left hand side actually also blows up, and the QES formula loses its meaning.
\end{itemize}

In order to solve these puzzles, it seems that the crucial issue is to understand how to deal with the finiteness of Newton's constant on the boundary, as well as with the UV cutoff of the bulk fields. An interesting proposal in the case of the exterior region of a black hole, originally due to Susskind--Uglum \cite{Susskind_1994}, and developed in an important body of work (see for example \cite{Kabat:1995eq,Demers:1995dq,Jacobson:1994iw,Larsen:1995ax,Fursaev:1994ea}), is that these two issues are actually related. More precisely, the proposal is the following:  
\begin{conj}[Susskind--Uglum, \cite{Susskind_1994}]
The renormalization of the bulk entropy due to the bulk UV cutoff exactly cancels out the renormalization of Newton's constant in the area term.
\end{conj}

In other words, the choice of UV cutoff for the fields in the bulk is not independent of the running of Newton's constant, and the renormalization of one term of the QES formula in the bulk exactly cancels out the renormalization of the other term, making the right hand side of the QES formula only dependent on the value of Newton's constant (i.e. $N$) on the boundary, and not on the choice of UV cutoff in the bulk EFT. For $N$ large enough, the EFT description is then approximately valid, making it possible to obtain a well-defined, and cutoff-independent, QES formula. The way that Susskind--Uglum originally argued for this proposal is by resorting to EFT arguments and to the Euclidean path integral. These arguments are not fully rigorous and only valid in specific cases.

A more modern way to understand the QES formula, which reduces to the formula for black hole entropy considered by Susskind--Uglum in the case of a two-sided black hole, is through the lens of quantum error correction in AdS/CFT \cite{Almheiri:2014lwa,Harlow:2016vwg}. In this context, the bulk term of the QES formula is reinterpreted as the entanglement entropy of a state in a \textit{code subspace} of the full UV-complete Hilbert space. The area term then captures the entropy of the state which is \textit{not} associated with encoded observables, but instead, with other degrees of freedom in the CFT that do not appear in the bulk fields. The interpretation is that it is the entanglement of these UV degrees of freedom that conspires to create a geometry in the bulk and make spacetime emerge.

The early breakthroughs in holographic quantum error correction \cite{Pastawski_2015,Harlow:2016vwg,Cotler:2017erl,Hayden_2019} were mainly achieved in the context of toy models, that approximate the full-fledged holographic situation in terms of qubits and finite-dimensional Hilbert spaces. While these models already retain a lot of the important properties of holography, their discreteness makes it difficult to tackle the continuous nature of spacetime in the bulk EFT, the notion of large $N$ limit, and the infinite-dimensional nature of the boundary Hilbert space. However, recent progress \cite{Kang:2018xqy,Gesteau:2020rtg,Faulkner:2020hzi,Gesteau:2021jzp,FaulknerLi} has allowed to move past the finite-dimensional case by recasting holographic quantum error correction in the language of infinite-dimensional von Neumann algebras. In this more general language, it is now possible to precisely define the aforementioned notions. The goal of this paper is to show that this new language is enough to construct a framework in which the Susskind--Uglum conjecture can be formulated and proven.

The first task will be to understand how to obtain a formulation of holographic quantum error correction in the context of the large $N$ limit of AdS/CFT, in a way in which it is possible to derive an entropy formula. A first step towards this goal has been taken by Faulkner and Li \cite{FaulknerLi}, and in particular, it has been shown that it is possible to derive the JLMS formula and the correspondence between bulk and boundary modular flows \cite{Jafferis:2015del} in the large $N$ limit. 

The new input of this work will be to note that while the notions of convergence introduced in \cite{FaulknerLi} allow to derive results like the JLMS formula in the large $N$ limit, they are too loose to define a notion of code subspace that is robust enough to satisfy a formula that involves von Neumann entropies in the bulk, like the Ryu--Takayanagi formula. For this latter purpose, the large $N$ bulk von Neumann algebra instead needs to be regulated in order to isolate observables that contribute to the bulk entropy in the large $N$ limit. More precisely, I will construct a family of type $I$ von Neumann algebras which retain a finite amount of large $N$ bulk entropy. These algebras will be nested inside one another and related through conditional expectations, which will implement a renormalization group flow of code subspaces.

Once this new setup based on conditional expectations is introduced at the level of the bulk theory at large $N$, the next step will be to introduce the bulk-to-boundary maps, which, along the same lines as \cite{Akers_2022}, relate the semiclassical bulk theory to finite $N$ boundary theories. Under physically motivated assumptions, I will show that this newly introduced family of codes satisfies an asymptotic entropy formula when $N\rightarrow\infty$ on the boundary. An important ingredient will be the definition of the area term for an approximate quantum error-correcting code proposed in \cite{Akers_2022}, in terms of the entropy of the Choi--Jamiolkowski state associated to the code. Thanks to the regulation of the code subspace, each individual term of the formula will be well-defined. The conditional expectations in the bulk will concretely implement the renormalization group flow for bulk entropy, and the invariance of the entropy formula under this flow will yield a rigorous proof of the Susskind--Uglum conjecture. 

For clarity, the main findings of the paper are summarized below:

\begin{itemize}
    \item This work provides a construction of a setup in which holographic entropy formulae as well as a renormalization scheme for the code subspace can be rigorously defined in the large $N$ limit of holography, and the Susskind--Uglum conjecture can be proven.
    \item It clarifies some subtleties about the type of the bulk von Neumann algebras at large $N$. While the full, unregulated bulk algebra has type $II_\infty$ or $III_1$, the pertinent algebras contributing to the bulk entropy have type $I$ as long as this entropy is $O(1)$ in the large $N$ limit.
    \item The role of conditional expectations in holography will be clarified and further extended. In particular, it will be shown that they can implement the renormalization group flow for bulk degrees of freedom. This role is somewhat related to the original proposal of \cite{Faulkner:2020hzi} that the boundary-to-bulk map should be modelled by a conditional expectation. Here, I will argue that while this original picture breaks down at finite $N$ because reconstruction becomes approximate, it can still be made sense of at large $N$, and that the possible code subalgebras are related to each other inside the large $N$ algebra by conditional expectations that implement a renormalization group flow.
    \item This work also puts forward an intimate relationship between the Susskind--Uglum conjecture and the ER=EPR paradigm. In particular, the fact that the entropy formula is invariant under RG flow can be interpreted as the fact that depending on the choice of code subalgebra, some boundary entanglement can be seen either as entanglement in the code subspace or as a contribution to the area term without changing any of the physics. This shows a complete equivalence between entanglement and geometric contributions in this context.
\end{itemize}

The exposition is organized as follows: in Section \ref{sec:EFTHilbert}, the Hilbert space of large $N$ effective field theory in the bulk is introduced. After summarizing recent constructions related to large $N$ von Neumann algebras \cite{Leutheusser:2021frk,Leutheusser:2021qhd,Witten:2021unn} and asymptotically isometric codes \cite{FaulknerLi}, it is argued that the large $N$ von Neumann algebras need to be further regulated if one wants to be able to make sense of holographic entropy formulae. In Section \ref{sec:CodeRenormalize}, a precise notion of code subspace renormalization is constructed. UV-regulated bulk algebras and code subspaces at large $N$ are defined. The renormalization group flow between the algebras is implemented by conditional expectations. In Section \ref{sec:SusskindUglum}, this RG flow of code subspaces is mapped into the boundary theory in the case of finite-dimensional regulated algebras. It is shown that under physically relevant assumptions for the bulk-to-boundary map, an entropy formula is true for this family of codes, regardless of the choice of renormalization scale in the bulk. This provides an explicit proof of the Susskind--Uglum conjecture, and aligns with the ER=EPR proposal. In Section \ref{sec:generalize}, the results of the previous section are generalized to various more general cases, that amount to relaxing some conditions on the dimensions of the various algebras considered. Finally, Section \ref{sec:Discuss} comments on various potential extensions of this work and further directions.

\textbf{Technical warning:} Unless specified otherwise, all von Neumann algebras introduced here are hyperfinite and all Hilbert spaces are separable.

\section{The Hilbert space of effective field theory}
\label{sec:EFTHilbert}

Much of the work on the error-correcting structure of holography has been focusing on the subtle way in which the low-energy effective field theory in the bulk is encoded in the unitary boundary CFT. The increasingly precise interpretation based on quantum codes has proven to be very fruitful to understand delicate issues about the semiclassical limit of gravity, such as the consistency between black hole evaporation and unitarity \cite{Penington:2019npb,https://doi.org/10.48550/arxiv.2207.06536}. This section will review how bulk effective field theory emerges in the large $N$ limit of AdS/CFT, and the error-correcting properties of the mappings of the $N=\infty$ theory (or perturbation theory around it) into large but finite $N$ theories. It will largely be based on the recent developments \cite{Leutheusser:2021frk,Leutheusser:2021qhd,Witten:2021unn,FaulknerLi}.

\subsection{Large $N$ algebras and the crossed product}

It has recently been argued by Leutheusser and Liu \cite{Leutheusser:2021frk,Leutheusser:2021qhd} (see also \cite{Leutheusser:2022bgi} for a more general version of this proposal) that in order to study the emergence of spacetime in the large $N$ limit of AdS/CFT, one needs to consider the von Neumann algebra generated by the single-trace operators of the gauge theory. More precisely, at large $N$, these operators behave like generalized free fields (i.e. their correlation functions factorize but they don't satisfy any equation of motion). One can then consider the GNS representation of the $C^\ast$-algebra of fields in a thermal state (i.e. the thermofield double Hilbert space), and define the corresponding large $N$ algebra as the bicommutant of the GNS representation.

The main conjecture of \cite{Leutheusser:2021frk,Leutheusser:2021qhd} is that the large $N$ von Neumann algebra changes type across the Hawking--Page transition: if $T_{HP}$ is the Hawking--Page temperature, for $T<T_{HP}$, it has type $I$, and for $T>T_{HP}$, it has type $III_1$. The change of type of the large $N$ algebra is associated to the presence of a continuous Källén--Lehmann density, which, in turn, can be related to the emergence of a black hole horizon, the connectedness of the thermofield double state, and the lack of factorizability of the large $N$ Hilbert space. 

It is possible to take this idea one step further. The issue with type $III_1$ algebras is that they do not admit any faithful normal semifinite trace, so it is not possible to define a good notion of entropy on them. It is therefore difficult to express the usual holographic statements that involve entanglement entropy in the bulk. What was shown in \cite{Witten:2021unn} is that adding perturbative $\frac{1}{N}$ corrections and quantizing an extra mode corresponding to the ADM mass of spacetime amounts to deforming the bulk algebra from a type $III_1$ factor to a type $II_\infty$ factor through a standard construction known as the crossed product with the modular automorphism group. Such a construction has proven to be very important in pure mathematics, in the context of the classification of type $III$ factors \cite{Connes:1973}.

In a type $II_\infty$ factor, it is possible to define a one-parameter family of traces, but there is no canonical choice of normalization. Instead, the traces are related to one another by a scaling automorphism. The consequence is that in these algebras, von Neumann entropy is only defined up to an overall constant \cite{Witten:2021unn}. The physical meaning of this is that entropy in a type $II_\infty$ factor is a renormalized version of entropy, where an infinite amount of entanglement has been thrown away. How much entanglement needs to be thrown away is arbitrary and cannot be fixed simply by looking at the type $II_\infty$ factor. This is in sharp contrast with the situation in a type $I$ algebra, which corresponds to the UV-complete description of the algebra on one side of the thermofield double (in AdS/CFT, at finite $N$ any CFT algebra has type $I$). In a type $I$ algebra, the trace is uniquely defined and there is only one way to define entropy, because no infinity needs to be subtracted. What this tells us is that entropy in the crossed product is a coarse-grained quantity that can only be calculated up to an overall constant, whereas entropy in the UV-finite theory is of course uniquely defined.

In \cite{https://doi.org/10.48550/arxiv.2209.10454}, it was shown that entropy in the large $N$ crossed product algebra can be identified with the generalized entropy of a quantum extremal surface, up to an overall constant. The argument requires a formula \cite{Wall:2011hj} relating the generalized entropy at late times to the one on a given time slice. This approach gives a justification for the quantum extremal surface formula directly at the level of the bulk theory in the large $N$ limit (including an extra mode corresponding to the ADM mass). The focus of this paper is different, as it will derive the holographic entropy formula from the perspective of quantum error correction in the large $N$ limit of holographic codes.

\subsection{Embedding into UV-complete theories and quantum error correction}

The perspective here will be to derive a family of Ryu--Takayanagi formulae in the large $N$ limit of holography within the framework of quantum error-correction. A quantum code is essentially the data of a code Hilbert space $\mathcal{H}_{code}$, a boundary Hilbert space $\mathcal{H}_{phys}$, and a bulk-to-boundary map $V:\mathcal{H}_{code}\longrightarrow\mathcal{H}_{phys}$. 

Recently, Faulkner and Li \cite{FaulknerLi} observed that in order to study the large $N$ limit of holography, one does not need to consider a single code, but rather, a sequence of codes. More precisely, if $\mathcal{H}_{code}$ is the bulk Hilbert space at large $N$ (i.e. the Hilbert space on which the large $N$ algebra is represented), then there exist an infinity of boundary Hilbert spaces $\mathcal{H}_{N}$ and bulk to boundary maps $V_N$: one for each choice of $N$. These maps are required to be ``asymptotically isometric": they are required to approach the properties of an isometry, so that an exact quantum error-correcting structure is recovered as $N$ goes to infinity.  In particular, Faulkner and Li impose the following conditions:
\begin{align}
    V_N^\dagger V_N- Id \underset{N\rightarrow\infty,\,w.o.t.}{\longrightarrow}0,
\end{align}
and for all $A\in M_{code}$, where $M_{code}$ is the algebra to reconstruct,
\begin{align}
    \gamma_N(A)V_N- V_NA \underset{N\rightarrow\infty,\,s.o.t.}{\longrightarrow}0,
\end{align}
where $\gamma_N$ is the reconstruction map and the subscripts w.o.t. and s.o.t. indicate that the limits are respectively taken for the \textit{weak} and \textit{strong} operator topologies.

The weak and strong operator topologies does not coincide with the topology induced by the operator norm in infinite dimensions. While this might seem like a technical subtlety, this choice of topology actually carries a lot of physical meaning. The two convergences described above mean that if one \textit{fixes} two states $\ket{\psi}$ and $\ket{\varphi}$ in $\mathcal{H}_{code}$, then the matrix elements $\bra{\psi}V_N^\dagger V_N\ket{\varphi}$ converge towards $\braket{\psi|\varphi}$, and the vectors $(\gamma_N(A)V_N- V_NA)\ket{\psi}$ converge to zero. However, this convergence is \textit{not} required to be uniform in the choice of $\ket{\psi}$ and $\ket{\varphi}$, or $A$. 

This means that for each value of $N$, no matter how large, there may (and do!) exist states for which reconstruction fails dramatically. This makes sense physically: if $N$ is very large but fixed, the effective field theory picture will break down for some operators that scale parametrically with $N$, and the reconstruction map should not be trusted anymore for such high energy excitations. 

This is a first hint that if one wants to relate the entropy of a large but finite $N$ boundary theory to that of the bulk effective field theory at large $N$ like the Ryu--Takayanagi formula does, the full bulk von Neumann algebra may not be considered as a code subalgebra. Rather, this algebra needs to be \textit{regulated}, and only a subset of observables for which one can require stronger reconstruction properties, must be singled out. In other words, {\it the large $N$ algebra must be renormalized.} This introduces an arbitrary choice of regularization. The next step is therefore to introduce a general method to regulate the code subspace of a large $N$ code, and to introduce a renormalization group flow between different choices of regulations.

\section{Code subspace renormalization}
\label{sec:CodeRenormalize}

It is now clear that no matter how large $N$ is chosen in the UV-complete boundary theory, the mapping of the bulk effective field theory into the boundary theory will dramatically fail for some operators. It is then useful to define UV-regulated algebras of observables in the bulk, that can be mapped into the boundary theory with good precision for large enough, but fixed, values of $N$. In this section, I introduce a new notion of \textit{code subspace renormalization}, that allows to define such algebras and compare their respective cutoff scales.

\subsection{What should a consistent renormalization procedure be?}

The first goal is to study which possible von Neumann subalgebras of the large $N$ algebra of observables can be chosen to match the bulk entropy term of the Ryu--Takayanagi formula. At least as a first step, I will assume the standard scaling for bulk entropy 
\begin{align}
    S_{bulk}=O(1),
\end{align}
in the large $N$ limit. It turns out that the fact that a normal state on a von Neumann algebra carries a finite amount of entropy strongly constrains its possible type. Actually, it will be shown in Appendix \ref{sec:BoundedTypeI} that a very mild finiteness of entropy condition implies that the associated von Neumann algebra must have type $I$. This can already be noticed from an intuitive point of view: in the type $II$ and type $III$ cases, divergences of entropies imply that von Neumann entropy is either not possible to define at all, or that the only possible equivalent throws away an infinite amount of entanglement. Hence, it is natural, at least in the case in which one requires the bulk entropy term to be uniformly bounded in $N$, to associate it to a type $I$ von Neumann subalgebra of the large $N$ observables. Note that this is a somewhat more abstract argument that shows that it is not right to consider the whole large $N$ algebra of observables in order to show something like a Ryu--Takayanagi formula.\footnote{If the bulk entropy were to diverge at large $N$, one would still expect to be able to single out a family of type $I$ subalgebras that carry an $N$-dependent amount of entropy.}

The crucial point of this paper is that such a choice of algebra, and Hilbert subspace on which it acts, is highly nonunique. There is a large amount of arbitrariness in how one chooses the subalgebra of the effective field theory that will be considered as the ``code subalgebra" of the holographic code.

However, one cannot choose the reconstructible subalgebra and Hilbert subspace completely arbitrarily. In order to show Ryu--Takayanagi formulae, it is indispensable that a good notion of complementary recovery remains. This implies that the Hilbert subspaces of the large $N$ Hilbert space, and the von Neumann subalgebras, must be chosen so that the commutant structures are compatible with the one of the full large $N$ Hilbert space, and with one another.

The goal of the rest of this section will be to introduce a setup that defines such possible choices of type $I$ subalgebras and relates them with each other in a way that preserves complementarity. The idea is to construct a renormalization scheme by considering a nested family of type $I$ factors that can be projected onto one another by ``integrating out some entanglement". These successive coarse-graining operations can be interpreted as implementing a renormalization group flow. In an operator-algebraic setting, projections of norm one are called conditional expectations. More precisely:\footnote{In some simple cases which will turn out to be relevant here, conditional expectations can be expressed in a more explicit manner, see for example the discussion below Proposition \ref{prop:cexptypei}.}

\begin{defn}
Let $N\subset M$ be an inclusion of von Neumann algebras. A conditional expectation $E:M\rightarrow N$ is a linear map such that $E(Id)=Id$, and for $n_1,n_2\in N$ and $m\in M$, \begin{align}E(n_1mn_2)=n_1E(m)n_2.\end{align}
\end{defn}

Hence, the right structure to look at is a family of conditional expectations $\mathcal{E}_\lambda$, that project the observables of the large $N$ theory $M$ onto some type $I$ subalgebra $M_\lambda$, for $\lambda$ with values in a partially ordered set. The order in this set should be understood as a fine graining direction, so I will assume that for every $\mu\leq\lambda$, there exists a faithful normal conditional expectation $E_{\lambda\mu}:M_\lambda\rightarrow M_\mu$. This family of conditional expectations implements the renormalization group flow of the holographic code. It then turns out that conditional expectations react very well with commutant structures. This is essentially the content of Takesaki's theorem \cite{TAKESAKI1972306}. The idea will be to construct Hilbert spaces of states that are invariant under the successive conditional expectations, and Takesaki's theorem will guarantee that these subspaces are compatible with the commutant structure.

Interestingly, the structure of conditional expectation has already been introduced as a model of exact holographic codes in the past \cite{Faulkner:2020hzi}. In retrospect, this is not surprising, and reflects the fact that exact entanglement wedge reconstruction is recovered (under this reinterpretation in terms of renormalization) in the large $N$ theory. Note that the link between conditional expectations and renormalization group flow has also already been mentioned in \cite{Furuya:2020tzv}.

\subsection{Formal setup}

More formally, the ideas presented above can be captured by the following definition of a code subspace renormalization scheme:\footnote{Assumptions relative to faithfulness made throughout the paper are here essentially for convenience. The von Neumann algebras considered here are also all taken to be factors for simplicity.}

\begin{defn}
Let $M$ be a von Neumann factor. A \textit{code subspace renormalization scheme} for $M$ is a datum $(\Lambda,(M_\lambda)_{\lambda\in\Lambda},(\mathcal{E}_\lambda)_{\lambda\in \Lambda},(E_{\lambda\mu})_{\lambda,\mu\in \Lambda,\;\lambda\geq \mu},\omega)$, where:
\begin{itemize}
\item $\Lambda$ is a partially ordered set,
\item The $M_\lambda$ are type $I$ subfactors of $M$,
\item The $\mathcal{E}_\lambda$ are faithful normal conditional expectations from $M$ onto $M_\lambda$,
\item The $E_{\lambda\mu}$ are faithful normal conditional expectations from $M_\lambda$ onto $M_\mu$.
\item For $\lambda\geq\mu\geq\nu$, the following compatibility relations hold:
\begin{align}
    \mathcal{E}_\mu=E_{\lambda\mu}\circ\mathcal{E}_\lambda,
\end{align}
\begin{align}
    E_{\lambda\nu}=E_{\mu\nu}\circ E_{\lambda\mu}.
\end{align}
\item $\omega$ is a faithful normal state on $M$ that is invariant under all the $\mathcal{E}_\lambda$.
\end{itemize}
\end{defn}

 The next step is to introduce a Hilbert space of states on which $M$ acts. A natural choice, if one wants to think of a situation in which there are two boundary CFT's with a black hole in the center, is to think of the GNS representation of $M$ in the state $\omega$, which carries a nontrivial commutant structure for $M$:

\begin{defn}
Let $(\Lambda,(M_\lambda)_{\lambda\in\Lambda},(\mathcal{E}_\lambda)_{\lambda\in \Lambda},(E_{\lambda\mu})_{\lambda,\mu\in \Lambda,\;\lambda\geq \mu},\omega)$ be a code subspace renormalization scheme. The \textit{unregulated code subspace} associated to $(\Lambda,(M_\lambda)_{\lambda\in\Lambda},(\mathcal{E}_\lambda)_{\lambda\in \Lambda},(E_{\lambda\mu})_{\lambda,\mu\in \Lambda,\;\lambda\geq \mu},\omega)$ is the GNS Hilbert space $\mathcal{H}$ of $M$ in the state $\omega$.
\end{defn}

Note that as an alternative to constructing the Hilbert space directly from the GNS procedure, one could also have defined the a code subspace renormalization scheme directly from the action of a von Neumann algebra on a Hilbert space that contains a cyclic separating vector whose restriction to the von Neumann algebra is invariant under the conditional expectations. By uniqueness of the GNS representation, such a construction is isomorphic to the one described above. The next step is to introduce Hilbert subspaces associated to the regulated subalgebras $M_\lambda$.\footnote{Related Hilbert subspaces were introduced in \cite{FaulknerLi} in the context of the proof of an asymptotic information-disturbance tradeoff, although compatibility with a conditional expectation structure was not assumed. For these subspaces, reconstruction assumptions involving finer than weak or strong operator topologies can be assumed, in the same spirit as what will be done in the next section of this paper.}

\begin{defn}
Let $(\Lambda,(M_\lambda)_{\lambda\in\Lambda},(\mathcal{E}_\lambda)_{\lambda\in \Lambda},(E_{\lambda\mu})_{\lambda,\mu\in \Lambda,\;\lambda\geq \mu},\omega)$ be a code subspace renormalization scheme. The \textit{regulated code subspaces} $\mathcal{H}_\lambda\subset\mathcal{H}$ associated to the von Neumann subalgebra $M_\lambda$ are the Hilbert spaces spanned by the $M_\lambda\ket{\Omega}$, where $\ket{\Omega}$ is the GNS vector associated to $\omega$ in $\mathcal{H}$.
\end{defn}

The following proposition makes it explicit why the structure just introduced is well-adapted to describe a renormalization scheme for entropy. In particular, it identifies Hilbert spaces associated to states that are invariant under the conditional expectations.

\begin{prop}
\label{prop:cexptypei}
Let $(\Lambda,(M_\lambda)_{\lambda\in\Lambda},(\mathcal{E}_\lambda)_{\lambda\in \Lambda},(E_{\lambda\mu})_{\lambda,\mu\in \Lambda,\;\lambda\geq \mu},\omega)$ be a code subspace renormalization scheme, and let $\lambda\geq\mu$. There exist decompositions of the form 
\begin{align}
    M=M_\lambda\otimes M_\lambda^c,
\end{align}
and 
\begin{align}
    M=M_\mu\otimes M_{\lambda\mu}\otimes M_\lambda^c,
\end{align}
where $M_\mu$ and $M_{\lambda\mu}$ are type $I$ factors. Moreover, the Hilbert spaces $\mathcal{H}_\lambda$ and $\mathcal{H}_\mu$ are isomorphic to the GNS Hilbert spaces of $M_\lambda$ and $M_\mu$ in the state $\omega$, and there exist Hilbert spaces $\mathcal{H}^0_\lambda$, $\mathcal{H}^0_\mu$, $\mathcal{H}^{0c}_\lambda$, $\mathcal{H}^{0c}_\mu$ such that 
\begin{align}
    \mathcal{H}=\mathcal{H}^0_\lambda\otimes \mathcal{H}_\lambda^{0c},
\end{align}
and
\begin{align}
    \mathcal{H}=\mathcal{H}^{0}_\mu\otimes\mathcal{H}^{0c}_{\mu}
\end{align}
and states $\ket{\chi_\mu}\in\mathcal{H}_{\mu}^{0c}$, $\ket{\chi_\lambda}\in\mathcal{H}_{\lambda}^{0c}$ such that 
\begin{align}
    \mathcal{H}_\lambda=\mathcal{H}^0_\lambda\otimes \ket{\chi_\lambda},
\end{align}
and
\begin{align}
    \mathcal{H}_\mu=\mathcal{H}^{0}_\mu\otimes\ket{\chi_\mu}.
\end{align}

Moreover there exists a further decomposition 
\begin{align}
\mathcal{H}_\mu^{0c}=\mathcal{H}^0_{\lambda\mu}\otimes\mathcal{H}_\lambda^{0c},
\end{align} and a state $\ket{\chi_{\lambda\mu}}\in\mathcal{H}^0_{\lambda\mu}$ under which
\begin{align}
    \ket{\chi_\mu}=\ket{\chi_{\lambda\mu}}\otimes\ket{\chi_\lambda}.
\end{align}
\end{prop}

Further, note that under the decomposition described by the above proposition, the conditional expectations take a very simple form (where the states have been identified with their expectation value functionals):
\begin{align}
\mathcal{E}_\lambda=Id_{\mathcal{B}(\mathcal{H}_\lambda)}\otimes\chi_\lambda, \quad\quad E_{\lambda\mu}=Id_{\mathcal{B}(\mathcal{H}_\mu)}\otimes\chi_{\lambda\mu}.
\end{align}

\begin{proof}
The first two factorizations follow from the fact that the $M_\lambda$ are type $I$ factors, for a proof see paragraph 9.15 of \cite{Stratila_2020}.
With these factorizations in hand, Equation 4.10 of \cite{Faulkner:2020hzi} guarantees that $\omega$, which is an invariant state under both $\mathcal{E}_\lambda$ and $E_{\lambda\mu}$ must have the form $\omega_\mu\otimes\omega_{\lambda\mu}\otimes\omega^c_\lambda.$ Hence its GNS vector has the form 
\begin{align}
\ket{\Omega}=\ket{\chi_\mu}\otimes\ket{\chi_{\lambda\mu}}\otimes\ket{\chi^c_{\lambda}}.
\end{align}
The result straightforwardly follows.
\end{proof}

The above factorizations make it possible to calculate von Neumann entropy very explicitly for invariant states under the conditional expectations.

In the most explicit case in which it is possible to define a good additive notion of von Neumann entropy on $M$, for any state of the form $\ket{\psi_\lambda}\otimes\ket{\chi^c_\lambda}$, we have 
\begin{align}
    S(\ket{\psi_\lambda}\otimes\ket{\chi_\lambda}, M)=S(\ket{\chi_\lambda}, M_\lambda^c)+S(\ket{\psi_\lambda}, M_\lambda).
\end{align}
The entanglement entropy therefore splits into two pieces: a UV piece, $S(\ket{\chi_\lambda}, M_\lambda^c)$, which is generically divergent, and a (potentially) finite piece corresponding to the regulated type $I$ algebra $M_\lambda$. 

Moreover, under the renormalization group flow, we have the further decomposition
\begin{align}
\label{eq:renormalize}
    S(\ket{\psi_\mu}\otimes\ket{\chi_{\lambda\mu}}\otimes\ket{\chi_\lambda})=S(\ket{\chi_\lambda}, M_\lambda^c)+S(\ket{\chi_{\lambda\mu}}, M_{\lambda\mu})+S(\ket{\psi_\mu}, M_\mu).
\end{align}
The interpretation of the new term in the middle, $S(\ket{\psi_{\lambda\mu}}, M_{\lambda\mu})$, is that it integrates out some of the entropy associated to observables that are in $M_\lambda$ but not in $M_\mu$, and throws it into the UV piece of the entanglement of the state. This extra term will be reinterpreted as a renormalization term for Newton's constant in the next section. Note, however, that importantly, even in the case in which it is no longer possible to define the divergent term $S(\ket{\chi_\lambda}, M_\lambda^c)$ (or to give it a state counting interpretation), it is still possible to talk about the entropy of the type $I$ subalgebras involved, so that code subspace renormalization yields
\begin{align}
\label{eq:renormalizelite}
    S(\ket{\psi_\lambda}, M_\lambda)=S(\ket{\chi_{\lambda\mu}}, M_{\lambda\mu})+S(\ket{\psi_\mu}, M_\mu).
\end{align}
Crucially, it is only this latter equality that will be necessary to prove an entropy formula.

The other nice feature of the structure of code subspace renormalization scheme is that it respects the commutant structures. More precisely, by Takesaki's theorem, the modular structures of $M$, $M_\lambda$ and $M_\mu$ are compatible. This implies:

\begin{prop}
Let $(\Lambda,(M_\lambda)_{\lambda\in\Lambda},(\mathcal{E}_\lambda)_{\lambda\in \Lambda},(E_{\lambda\mu})_{\lambda,\mu\in \Lambda,\;\lambda\geq \mu},\omega)$ be a code subspace renormalization scheme. For $\lambda,\mu\in\Lambda$, in $\mathcal{H}_\lambda$ and $\mathcal{H}_\mu$,
\begin{align}
    M^\prime_\lambda=JM_\lambda J,\;\;\;\;M^\prime_\mu=JM_\mu J,
\end{align}
where $J$ is the modular conjugation of $\ket{\Omega}$ with respect to $M$.
\end{prop}

\begin{proof}
This is a direct consequence of Takesaki's theorem, given that states in $\mathcal{H}_\lambda$ and $\mathcal{H}_\mu$ are invariant under the corresponding conditional expectations.
\end{proof}

This fact guarantees the compatibility of the commutant structures along the renormalization group flow, and a nice nesting of all subspaces and subalgebras at hand. In particular, then, there also exist faithful normal conditional expectations $\mathcal{E}^\prime_\lambda$, $\mathcal{E}^\prime_\mu$ and $E^\prime_{\lambda\mu}$, defined on $M^\prime$ and $M^\prime_\lambda$ respectively, by 
\begin{align}
    \mathcal{E}^\prime_\lambda(X):=J\mathcal{E}_\lambda(JXJ)J,
\end{align}
\begin{align}
    \mathcal{E}^\prime_\mu(X):=J\mathcal{E}_\mu(JXJ)J,
\end{align}
\begin{align}
    E^\prime_{\lambda\mu}(X):=JE_{\lambda\mu}(JXJ)J.
\end{align}

Note that the compatibility condition 
\begin{align}
    \mathcal{E}^\prime_\mu=E^\prime_{\lambda\mu}\circ \mathcal{E}^\prime_\lambda
\end{align}
is satisfied.

Figure \ref{fig:cd} summarizes the structure of code subspace renormalization scheme, and how the compatibility between conditional expectations and commutant structures is realized, thanks to a commutative diagram. 

\begin{figure}
\centering
\[\begin{tikzcd}
	{M} && {M^\prime} && {\mathcal{H}} \\
	\\
	{M_\lambda} && {M_\lambda^\prime} && {\mathcal{H}_\lambda} \\
	\\
	{M_\mu} && {M_\mu^\prime} && {\mathcal{H}_\mu}
	\arrow["\prime", from=1-1, to=1-3]
	\arrow["{\mathcal{E}_\lambda}"', from=1-1, to=3-1]
	\arrow["{\mathcal{E}^\prime_\lambda}", from=1-3, to=3-3]
	\arrow["\prime", from=3-1, to=3-3]
	\arrow["{E_{\lambda\mu}}"', from=3-1, to=5-1]
	\arrow["{E^\prime_{\lambda\mu}}", from=3-3, to=5-3]
	\arrow["\prime", from=5-1, to=5-3]
	\arrow[from=1-5, to=3-5]
	\arrow[ from=3-5, to=5-5]
	\arrow["{\mathcal{E}_\mu}"', bend right=40, from=1-1, to=5-1]
	\arrow["{\mathcal{E}^\prime_\mu}", bend left=40, from=1-3, to=5-3]
\end{tikzcd}\]
\caption{A commutative diagram summarizing the structure of code subspace renormalization. Here the full bulk von Neumann algebras $M$ and $M^\prime$, which are commutants of each other, are mapped to the subalgebras $M_\lambda$ and $M_\mu$ and their commutants, corresponding to different cutoff scales, through the conditional expectations $\mathcal{E}_\lambda$, $\mathcal{E}_\mu$ and $\mathcal{E}^\prime_\lambda$, $\mathcal{E}^\prime_\mu$. The prime on the horizontal arrows denotes the commutant structure implemented by modular conjugation. Given that the states in $\mathcal{H}_\lambda$ and $\mathcal{H}_\mu$ are invariant under the conditional expectations, Takesaki's theorem guarantees that the commutant structure is respected, and that the diagram commutes.}
\label{fig:cd}
\end{figure}
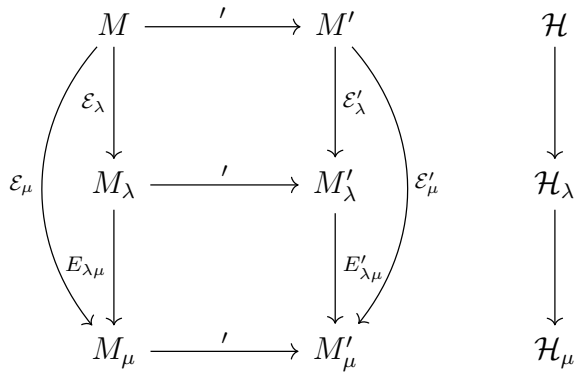

It follows from the previous analysis that the structure of code subspace renormalization scheme proposed here, and based upon nested type $I$ factors and Hilbert spaces related to each other by conditional expectations, is a good choice in the sense that it allows to completely decouple the contribution to the bulk entropy of different subalgebras of the large $N$ theory, and most importantly, to preserve complementarity. However, it does not provide a constructive way of defining these algebras - the most naive attempt of considering low energy products of single trace operators fails because such spaces are not closed under multiplication. 

Instead, one should think of the type $I$ factors introduced here as something closer to the type $I$ factors arising for subregions in theories that satisfy the split property. It has been argued in the past (see for example \cite{https://doi.org/10.48550/arxiv.1905.00577}) that such algebras can be thought of as UV-regulators for a quantum field theory. It is quite tempting to observe that restricting observables to a type $I$ factor can be thought of as imposing a ``brick wall" cutoff in the bulk QFT in the spirit of \cite{Susskind_1994}, and it would be very interesting to understand this better.

\section{A proof of the Susskind--Uglum conjecture}
\label{sec:SusskindUglum}
Now that a renormalization scheme for the bulk effective field theory has been defined, one can ask how the UV-regulated algebras map into the boundary theory. In this section, I show that for suitable values of $N$ and of the UV cutoff, the reconstruction map is good enough that the code satisfies a Ryu--Takayanagi formula. The value of the UV cutoff corresponds to a renormalization scale, and its choice is entirely arbitrary as long as it remains within a suitable range. I show that the Ryu--Takayanagi formula is invariant under the renormalization group flow. This provides an explicit proof of the Susskind--Uglum conjecture.

\subsection{The bulk to boundary map}

\begin{figure}
\centering
\includegraphics[height=5cm,width=10cm]{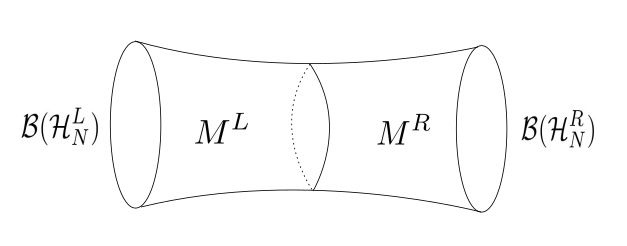}
\caption{The code in the case of two entangled CFT's on a compact space. The large $N$ algebras $M^L$ and $M^R$ need to be regulated in order for the map to the finite $N$ algebras $\mathcal{B}(\mathcal{H}_N^L)$ and $\mathcal{B}(\mathcal{H}_N^R)$ to allow the derivation of an entropy formula.}
\label{fig:code}
\end{figure}

The next step in this work is to map the effective theory at large $N$ in the bulk to a finite $N$ theory on the boundary. In order to do this, one needs to introduce one more object: the bulk-to-boundary map. This motivates the following definition of a renormalizable large $N$ quantum error-correcting code, see Figure \ref{fig:code}.\footnote{Once again note that here I will be mainly focusing on the case of two entangled CFT's, rather than the case of subregions of one CFT. This is because the latter case would require an extra regulation procedure due to the infiniteness of the Ryu--Takayanagi surface, but this regulation procedure would not teach us anything meaningful about the physics described here.}
\begin{defn}
A \textit{renormalizable large-}$N$ \textit{holographic quantum error correcting code} is defined by the data of a sequence of Hilbert spaces $(\mathcal{H}_N)_{N\in\mathbb{N}}$, and a sequence of contractive maps $V_N:\mathcal{H}\longrightarrow\mathcal{H}_N\otimes\mathcal{H}_N$, where $\mathcal{H}$ corresponds to the unregulated Hilbert space of a code subspace renormalization scheme $(\Lambda,(M_\lambda)_{\lambda\in\Lambda},(\mathcal{E}_\lambda)_{\lambda\in \Lambda}, (E_{\lambda\mu})_{\lambda,\mu\in \Lambda,\;\lambda\geq \mu}, \omega)$. 
\end{defn}

In order to explicitly differentiate between the two sides of the code, the algebras $M_\lambda$ will often be denoted $M_\lambda^L$, and their commutants $M_\lambda^R$, and similarly for the Hilbert spaces $\mathcal{H}_N^L$ and $\mathcal{H}_N^R$ on the boundary.

It was argued above that for each finite value of $N$, the holographic code drastically fails to map some of the bulk states to the boundary with good precision. Hence, one needs to renormalize the code subspace. The idea of this section will be to use the framework of code subspace renormalization put forward above to define UV-regulated subalgebras for which strong enough reconstruction properties can be imposed. A Ryu--Takayanagi formula will then be proven. Moreover, under the renormalization group flow, I will explicitly show that the corrections to the area term and the bulk entropy term of the formula exactly compensate each other.

A technical remark is that in order to be able to control the von Neumann entropy of states on the boundary, it will be necessary to impose strong reconstruction assumptions on the states in regularized code subspaces (typically, a nonperturbatively small error in the $1/N$ expansion, or at least, small enough that the polynomially divergent factor in Fannes' inequality doesn't spoil the conclusions). In order to be able to impose such an assumption, one needs to include backreaction effects in the code subspace that go beyond the strict large-$N$ limit of Leutheusser--Liu, and may introduce some $N$-dependence not only at the level of the boundary theory, but also of the code subspace and its renormalization scheme. Different ideas exist to construct code subalgebras allowing for perturbation theory in $1/N$ \cite{Witten:2021unn,FaulknerLi}. While it is beyond the scope of this paper to attempt such a construction, I emphasize that the results introduced here are still valid if the code subspace and its renormalization scheme depend perturbatively on $N$ (and it would be interesting to find a systematic way to interpolate between schemes with different values of $N$). For convenience, the rest of this section (and of this paper) will not make it explicit in its notations that the code subspace renormalization scheme and its Hilbert spaces themselves might depend on $N$.

\subsection{A Ryu--Takayanagi formula}

This section will cover the simplest case in which the chosen bulk algebra $M_\lambda$, as well as the boundary algebra, are taken to be finite-dimensional. In this case, the Ryu--Takayanagi formula can be derived by adapting the proof of a result of \cite{Akers_2022}, which introduces a framework in which it is possible to talk about area terms for approximate and non-isometric codes. It generalizes to the approximate case the notion of area term for quantum codes introduced in \cite{Harlow:2016vwg}.

Let us first briefly recall the setup of \cite{Akers_2022}, and especially, how one defines a good notion of area term from the structure of the code in this context. The idea is to consider a map 
\begin{align}
    V:\mathcal{K}_L\otimes\mathcal{K}_R\longrightarrow\mathcal{H}_L\otimes\mathcal{H}_R,
\end{align} 
where $\mathcal{K}_{L,R}$ are left and right bulk Hilbert spaces (in the context considered here where the bulk regions are associated to different sides of the black hole) and are finite-dimensional, and introduce a reference Hilbert space 
\begin{align}
    \mathcal{K}^{ref}:=\mathcal{K}^{ref}_L\otimes\mathcal{K}^{ref}_R,
\end{align} 
where $\mathcal{K}^{ref}_L$ and $\mathcal{K}^{ref}_R$ have the same dimensions as $\mathcal{K}_L$ and $\mathcal{K}_R$ and are identified with dual Hilbert spaces.
One can then introduce the canonical maximally entangled state 
\begin{align}
    \ket{MAX}\in\mathcal{K}_L\otimes\mathcal{K}_R\otimes\mathcal{K}^{ref},
\end{align} 
which maximally entangles $\mathcal{K}_{L,R}$ and $\mathcal{K}^{ref}_{L,R}$ together in a basis-independent manner.
The Choi--Jamiolkowski state is then 
\begin{align}
    \ket{CJ}:=(V\otimes Id)\ket{MAX}\in\mathcal{H}_L\otimes\mathcal{H}_R\otimes\mathcal{K}^{ref}.
\end{align} 
The definition of area term for a subregion proposed in \cite{Akers_2022} (for example, associated to $\mathcal{H}_{L}$) is then given by:

\begin{defn}
In the code defined by the map $V$, the area term of the region $L$ is defined by
\begin{align}
    A(\mathcal{K}_{L}):=S(\ket{CJ},\mathcal{B}(\mathcal{H}_L)\otimes\mathcal{B}(\mathcal{K}^{ref}_L)).
\end{align} 
\end{defn}

Note that this definition is independent of the choice of state in $\mathcal{K}_L\otimes\mathcal{K}_R$.\footnote{In the context of this paper, $\mathcal{K}_L$ and $\mathcal{K}_R$ are the regulated Hilbert spaces, and this independence reflects the fact that the regulated algebras are assumed to not have a center.} However, it \textit{does} depend on the choice of code subspace $\mathcal{K}_L\otimes\mathcal{K}_R$. This point was not made explicit in \cite{Akers_2022} as in their context, only one code subspace is considered. Here however, in the novel framework of code subspace renormalization, the choice of code subspace and code subalgebras becomes highly nonunique, and it will turn out to be very important that the value of the area term does depend on the choice of code subspace, even though it does not depend on the choice of a particular state inside it.

The proof of the Ryu--Takayanagi formula in \cite{Akers_2022} heavily relies on the Fannes inequality, which turns out to become vacuous in the case of an infinite-dimensional boundary Hilbert space. In conformal field theory, even at finite $N$, the Hilbert space is infinite-dimensional, and the boundary algebra is a type $I_\infty$ factor. However it is usual to assume, as a first approximation, that the logarithm of the dimension of $\mathcal{H}_N$ grows polynomially in $N$ in order to obtain a proof of the Ryu--Takayanagi formula. This is usually justified due to the fact that the boundary entropy grows polynomially in $N$. In Section \ref{sec:generalize}, I will introduce an alternative to this assumption, thanks to the use of an alternative to Fannes' inequality due to Winter \cite{Winter_2016}. However I will simply make a finite-dimensional assumption in this section, as the proof will be less technical and easier to follow for the reader who is content with such a simplification.

Before stating and proving the theorem, let us introduce the statement of Fannes' inequality for convenience (see for example \cite{Akers_2022}). 

\begin{prop}
Let $\rho$ and $\sigma$ be subnormalized density matrices on a Hilbert space $\mathcal{H}$ of finite dimension $d$. Suppose that for $0<\varepsilon<e^{-1}$, $\|\rho-\sigma\|_1\leq\varepsilon.$ Then, 
\begin{align}
|S(\rho)-S(\sigma)|\leq\varepsilon\,\mathrm{log}\left(\frac{d}{\varepsilon}\right).
\end{align}
\end{prop}

The theorem then reads:

\begin{thm}
\label{thm:main}
Let $(\Lambda,(M_\lambda)_{\lambda\in\Lambda},(\mathcal{E}_\lambda)_{\lambda\in \Lambda},(E_{\lambda\mu})_{\lambda,\mu\in \Lambda,\;\lambda\geq \mu},\omega, V_N)$ be a large $N$ quantum error-correcting code. Suppose that for some choice of $\lambda$, $M_\lambda$ is finite-dimensional of dimension $d_\lambda^2$ constant in $N$. Suppose that there exists a polynomial $P$ in $N$ such that the dimension $d_N$ of $\mathcal{H}_N$ satisfies
\begin{align}
\mathrm{log}\,d_N\leq P(N).
\end{align}
Let $\ket{\Psi}\in\mathcal{H}_\lambda$ (normalized). Suppose that for all unitary operators $U^{L}_\lambda, \hat{U}^{L}_\lambda$ and $U^{R}_\lambda, \hat{U}^{R}_\lambda$ in $M^{L}_\lambda$ and $M^{R}_\lambda$, there exist unitary operators $\tilde{U}^{L}_\lambda$ and $\tilde{U}^{R}_\lambda$ (chosen in a measurable way) in $\mathcal{B}(\mathcal{H}_N^L)$ and $\mathcal{B}(\mathcal{H}_N^R)$ such that\footnote{The $\hat{U}$'s are introduced essentially to match with the notations of \cite{Akers_2022}. Also note that $\delta_N$ only needs to be smaller than all polynomials appearing in the proof.}
\begin{align}
\label{eq:reconstructionassumption}
    \|V_N U^{R}_\lambda U^{L}_\lambda\ket{\Psi}-\tilde{U}^{R}_\lambda \tilde{U}^{L}_\lambda V_N \hat{U}^{R}_\lambda \hat{U}^{L}_\lambda\ket{\Psi}\|\leq \delta_N,
\end{align}
where $\delta_N$ decays faster than any polynomial in $N$.
Then,
\begin{align}
   |S(\ket{\Psi},M^L_\lambda)+A(\mathcal{H}_{L,\lambda})-S(V_N\hat{U}^{R}_\lambda \hat{U}^{L}_\lambda\ket{\Psi},\mathcal{B}(\mathcal{H}^L_N))|\underset{N\rightarrow\infty}{\longrightarrow} 0.
\end{align}
\end{thm}

\begin{proof}
This theorem is essentially a translation of Theorem 5.5 of \cite{Akers_2022} in the present setup, and this proof will closely follow the strategy used there.

The Hilbert space $\mathcal{H}_\lambda$ (once identified with the $\mathcal{H}_\lambda^0$ of the previous section) can be written as $\mathcal{H}^L_\lambda\otimes\mathcal{H}^R_\lambda$, where the decomposition is consistent with $\mathcal{B}(\mathcal{H}_\lambda)=\mathcal{B}(\mathcal{H}^L_\lambda)\otimes\mathcal{B}(\mathcal{H}^R_\lambda)=M_\lambda^L\otimes M_\lambda^R$. Construct an isometry 
\begin{align}
    W^L_\lambda:&\mathcal{H}_\lambda\longrightarrow L^2(U(\mathcal{H}^L_\lambda))\otimes\mathcal{H}_\lambda\\\ket{\Psi}&\longmapsto\int dU_L\ket{U_L}_{U_L}\otimes U_L\ket{\Psi}.
\end{align}
One can use the Peter--Weyl theorem (see \cite{Akers_2022}, Lemma 4.4) to show that
\begin{align}
    W^L_\lambda\ket{\Psi}=\ket{\Psi}_{aR}\otimes\ket{MAX}_{Lr},
\end{align}
where $a$, $r$ are reference systems of dimensions equal to that of $\mathcal{H}^R_\lambda$ and $\mathcal{H}^L_\lambda$ and correspond to the fundamental and antifundamental representations of the unitary group. One defines the isometry $W^R_\lambda$ in an exactly similar way.
Then (see Equation 4.32 of \cite{Akers_2022}),
\begin{align}
    S(V_NW^L_\lambda W^R_\lambda\ket{\Psi},\,\mathcal{B}(\mathcal{H}^L_N\otimes\mathcal{H}_f\otimes\mathcal{H}^\ast_f))=A_L(\mathcal{H}^L_\lambda)+S(\ket{\Psi},\,\mathcal{B}(\mathcal{H}^L_\lambda)),
\end{align}
where $\mathcal{H}_f$ and $\mathcal{H}^\ast_f$ are the fundamental and antifundamental Hilbert spaces. Now, introduce the operators $W^{L,R}_N$, defined by
\begin{align}
\label{eq:wndef}
    W^{L,R}_N:=\int dU^{L,R}_\lambda \ket{U^{L,R}_\lambda}\otimes \tilde{U}^{L,R}_\lambda.
\end{align}
Substituting inequality \eqref{eq:reconstructionassumption}, one then obtains
\begin{align}
    \|W^L_NW^R_NV_N\hat{U}^L_\lambda\hat{U}^R_\lambda\ket{\Psi}-V_NW^L_\lambda W^R_\lambda\ket{\Psi}\|\leq\delta_N.
\end{align}
From this inequality, one can deduce a bound on the difference of entropies thanks to Fannes' inequality.\footnote{The bound derived in \cite{Akers_2022} misses some terms, and a few extra steps need to be taken care of in order to get a consistent bound. I am grateful to Chris Akers and Geoff Penington for clarifying this point.} First, recall that by the Peter--Weyl theorem, 
\begin{align}
L^2(U(\mathcal{H}^L_\lambda))=\underset{\mu}{\bigoplus}\,\mathcal{H}_\mu\otimes\mathcal{H}^\ast_\mu,
\end{align}
where $\mu$ runs over the finite-dimensional irreducible representations of $U(\mathcal{H}^L_\lambda)$. What is shown in Lemma 4.4 of \cite{Akers_2022} is that the maps $W^L_\lambda,\; W^R_\lambda$ only have range on the component of the direct sum corresponding to the fundamental representation of $U(\mathcal{H}^L_\lambda)$. Similarly, the image of $\mathcal{H}^L_N$ by $W^L_N$ sits inside $\mathcal{H}^L_N\otimes L^2(U(\mathcal{H}^L_\lambda))$. This fact makes it natural to introduce the Hilbert space
\begin{align}\tilde{\mathcal{H}}^L_N:=\mathcal{H}_f\otimes\mathcal{H}^\ast_f\otimes\mathcal{H}^L_N+W^L_N\mathcal{H}^L_N,\end{align}
where the subscript $f$ denotes the fundamental representation (the sum need not be direct), and similarly the Hilbert space
\begin{align}\tilde{\mathcal{H}}^R_N:=\mathcal{H}_f\otimes\mathcal{H}^\ast_f\otimes\mathcal{H}^R_N+W^R_N\mathcal{H}^R_N.\end{align}
Note that 
\begin{align}
\mathrm{dim}(\tilde{\mathcal{H}}^L_N)\leq d_N(d_\lambda^2+1).
\end{align}
Moreover, states of both forms $W^L_NW^R_NV_N\hat{U}^L_\lambda\hat{U}^R_\lambda\ket{\Psi}$ and $V_NW^L_\lambda W^R_\lambda\ket{\Psi}$ are in $\tilde{\mathcal{H}}^L_N\otimes\tilde{\mathcal{H}}^R_N$. As the $W$'s are isometries, one can also deduce:
\begin{align}
S(V_N\hat{U}^L_\lambda\hat{U}^R_\lambda\ket{\Psi},\mathcal{B}(\mathcal{H}^L_N))=S(W^L_NW^R_NV_N\hat{U}^L_\lambda\hat{U}^R_\lambda\ket{\Psi},\mathcal{B}(W^L_N\mathcal{H}^L_N))=S(W^L_NW^R_NV_N\hat{U}^L_\lambda\hat{U}^R_\lambda\ket{\Psi},\mathcal{B}(\tilde{\mathcal{H}}^L_N)).
\end{align}
On the other hand,
\begin{align}
S(V_NW^L_\lambda W^R_\lambda\ket{\Psi},\mathcal{B}(\mathcal{H}^L_N\otimes\mathcal{H}_f\otimes\mathcal{H}^\ast_{f}))=S(V_NW^L_\lambda W^R_\lambda\ket{\Psi},\mathcal{B}(\tilde{\mathcal{H}}^L_N)).
\end{align}
Hence one can apply Fannes' inequality the density matrices associated to these two subnormalized states on $\mathcal{B}(\tilde{\mathcal{H}}^L_N)$, whose 1-norm distance is smaller than $2\delta_N$ by Lemma C.1 of \cite{Akers_2022}. This implies (for $N$ large enough):
\begin{align}
|S(\ket{\Psi},M^L_\lambda)+A(\mathcal{H}_{L,\lambda})-S(V_N\hat{U}^{R}_\lambda \hat{U}^{L}_\lambda\ket{\Psi},\mathcal{B}(\mathcal{H}^L_N))|\leq 2\delta_N\mathrm{log}\left(\frac{d_N(d_\lambda^2+1)}{2\delta_N}\right).
\end{align}
As $\mathrm{log}\,d_N$ grows at most polynomially, $d_\lambda$ is fixed here (see Section \ref{sec:generalize} for a setup in which this assumption is relaxed), and $\delta_N$ decays faster than polynomially, this expression goes to zero at large $N$.
\end{proof}

Note that although the dimension of the code subspace was kept fixed here, it can also be made to grow with $N$ as long as it grows slowly enough that issues related to state-dependence and entanglement wedge jumps do not arise. In this more complicated case, a setup more akin to the product unitary reconstruction assumption proposed in \cite{Akers_2022} must be used instead. These cases will be further discussed in the next section.

Another remark can be made at this stage: it has been pointed out \cite{Witten:2021unn,Chandrasekaran:2022eqq} that in the case where the bulk algebra is a type $II_\infty$ factor, an entanglement entropy associated to the entire, unregulated bulk algebra can be defined but only up to an overall constant, which captures the fact that an infinite amount of entanglement needs to be thrown out in order to obtain a finite answer, and that the part that remains is arbitrary. Here, it is tempting to choose this constant in an $N$-dependent fashion so that the entropy of the bulk state in the unregulated algebra matches $S(V_N\hat{U}^{R}_\lambda \hat{U}^{L}_\lambda\ket{\Psi},\mathcal{B}(\mathcal{H}^L_N))$. Then, the resulting statement can be seen as another instance of the fact that bulk generalized entropy equals boundary entropy \cite{Chandrasekaran:2022eqq}.

\subsection{Invariance under renormalization group flow}

Now that it has been shown that the Ryu--Takayanagi formula is true for \textit{any} choice of finite-dimensional $M_\lambda$, the previous result, coupled to the framework of code subspace renormalization introduced above, allows to very explicitly demonstrate the validity of the Susskind--Uglum conjecture, and to isolate the counterterm that gets reabsorbed into the area term under the renormalization group flow.

\begin{thm}
\label{thm:SusskindUglum}
Consider two choices of code subspace regularization $M_\lambda$ and $M_\mu$, such that $\lambda\geq\mu$, and $\ket{\Psi}\in\mathcal{H}_\mu$ (normalized). Then, under the assumptions of Theorem \ref{thm:main} on $M_\lambda$, for all unitaries $\hat{U}^L_\mu,\,\hat{U}^R_\mu$,
\begin{align}
   |S(\ket{\Psi},M_\mu)+A(\mathcal{H}_{L,\mu})-S(V_N\hat{U}^L_\mu\hat{U}^R_\mu\ket{\Psi},\mathcal{B}(\mathcal{H}_N))|\underset{N\rightarrow\infty}{\longrightarrow} 0,
\end{align}
and
\begin{align}
   |S(\ket{\Psi},M_\mu)+S(\ket{\Psi},M_{\lambda\mu})+A(\mathcal{H}_{L,\lambda})-S(V_N\hat{U}^L_\mu\hat{U}^R_\mu\ket{\Psi},\mathcal{B}(\mathcal{H}_N))|\underset{N\rightarrow\infty}{\longrightarrow} 0.
\end{align}
In particular,
\begin{align}
   |A(\mathcal{H}_{L,\mu})-(S(\ket{\Psi},M_{\lambda\mu})+A(\mathcal{H}_{L,\lambda}))|\underset{N\rightarrow\infty}{\longrightarrow} 0.
\end{align}
\end{thm}
\begin{proof}
The proof is straightforward. From the previous result, 
\begin{align}
   |S(\ket{\Psi},M_\mu)+A(\mathcal{H}_{L,\mu})-S(V_N\hat{U}^L_\mu\hat{U}^R_\mu\ket{\Psi},\mathcal{B}(\mathcal{H}_N))|\underset{N\rightarrow\infty}{\longrightarrow} 0,
\end{align}
and
\begin{align}
   |S(\ket{\Psi},M_\lambda)+A(\mathcal{H}_{L,\lambda})-S(V_N\hat{U}^L_\mu\hat{U}^R_\mu\ket{\Psi},\mathcal{B}(\mathcal{H}_N))|\underset{N\rightarrow\infty}{\longrightarrow} 0.
\end{align}
The properties of code subspace renormalization (in particular, Equation \eqref{eq:renormalizelite}) imply the identity
\begin{align}
\label{eq:equivalence}
    S(\ket{\Psi},M_\lambda)=S(\ket{\Psi},M_\mu)+S(\ket{\Psi},M_{\lambda\mu}),
\end{align}
which immediately yields the result.
\end{proof}
The physical meaning of \eqref{eq:equivalence} is that the sum of the area term associated to $M_\lambda$ and of the entanglement entropy contained in $M_{\lambda\mu}$ gives the area term associated to $M_\mu$. In other words, under code subspace renormalization, what was previously accounted for in the bulk entropy term now becomes part of the area term associated to $M_\lambda$. This is exactly the Susskind--Uglum prediction! Therefore, Theorem \ref{thm:SusskindUglum} can be seen as a rigorous statement of the Susskind--Uglum conjecture for the above choice of code subspace renormalization scheme.

\subsection{Susskind--Uglum as ER=EPR}

It is now established that Theorem \ref{thm:SusskindUglum} provides a rigorous proof of the Susskind--Uglum conjecture. I will now argue that, on top of providing a proof, it also implies a \textit{reinterpretation} of this conjecture in terms of the ER=EPR proposal. 

The crucial point is that Theorem \ref{thm:main} is valid for \textit{all} values of $\lambda$. This implies that when $\lambda$ increases, the amount of information contained in the area term decreases, whereas the amount of information contained in $M_\lambda$ increases. This is possible because the Choi--Jamiolkowski state depends on the choice of code subspace. In particular, the Choi--Jamiolkowski state associated to a larger code subspace will be associated to a smaller area term than the one associated to a smaller code subspace, and will not correspond to a state on the smaller code subspace.

What does this mean physically? Recall that in this paper's approach (just like in that of \cite{Akers_2022}), the entanglement structure of the Choi--Jamiolkowski state \textit{defines} the area term (including its quotienting by $4G_N$). The arbitrariness in the choice of $\lambda$ means that some amount of entropy contained in the area term of the Ryu--Takayanagi formula for a given choice of cutoff $\lambda$ can equivalently be seen as being part of the code subspace entropy for another choice of $\lambda$. This means that some of the entanglement of the boundary state can equivalently be interpreted as bulk entanglement, or as a contribution to the area term. This concretely equates an entanglement quantity to a contribution to a geometric term. 

This type of equivalence between entanglement and geometry falls into the general paradigm of ER=EPR \cite{Maldacena_2013}. Here, it is the choice of renormalization scale, which is completely arbitrary, that underlies this equivalence. As a result, the theorem proven in this paper can be seen as extra evidence for the fact that entanglement and geometry are two sides of the same coin in quantum gravity.

\section{Generalizations}
\label{sec:generalize}
The proof of the previous section already carries all the essential physical ideas of this paper, and already demonstrates how all the salient features of the Susskind--Uglum conjecture can be derived from the structure of large $N$ quantum error correcting codes thanks to the notion of code subspace renormalization. However, some technical simplifications were made that it would be nice to lift. It turns out that trying to do so involves interesting mathematics. The goal of this rather technical section is to generalize theorems \ref{thm:main} and \ref{thm:SusskindUglum} to more involved setups. I start by introducing an infinite-dimensional analog of Fannes' inequality due to Winter \cite{Winter_2016}, for which the dimension of the Hilbert space gets replaced by an energy condition on the states, and use it to lift the finite-dimensional assumption on the boundary Hilbert space (which was motivated above by the finiteness of black hole entropy at finite $N$, but is still a simplification). Also, Theorem \ref{thm:main} assumes that one is picking a finite-dimensional renormalized code subspace, of constant dimension that does not grow with $N$. However, there are some contexts in which one would like to be able to make the dimension of the renormalized code subspace grow with $N$. In this section, I also describe possible generalizations of the result to these more complicated cases. These generalizations require interesting assumptions about the way in which the bulk theory is regulated, and it is an important problem to understand better how they can be implemented directly at the level of the large $N$ von Neumann algebras, along the lines of \cite{FaulknerLi}.

\subsection{Infinite-dimensional boundary at finite $N$}

In the finite-dimensional context of \cite{Akers_2022}, it was necessary to suppose that the logarithm of the dimension of $\mathcal{H}_N$ was polynomial in $N$ in order to obtain a proof of the Ryu--Takayanagi formula. This is because the proof requires an application of Fannes' inequality for the boundary Hilbert space, which explicitly involves its dimension. Resorting to dimension arguments is not fully valid for CFT Hilbert spaces, which are infinite-dimensional even for finite values of $N$. However, the energy of a thermal state of a CFT does not grow too fast in $N$, and I will show that this fact can be used as an alternative to the finite-dimensional argument of \cite{Akers_2022}. The key technical tool will be the machinery introduced in \cite{Winter_2016}, which states an analog of Fannes' inequality that involves an energy bound on the states rather than a dimension bound on the Hilbert space. More precisely, one can define: 

\begin{defn}
Let $\mathcal{B}(\mathcal{H})$ be a type $I$ factor. Let $H$ be a self-adjoint operator on $\mathcal{H}$ such that for all $\beta>0$, $e^{-\beta H}$ is trace-class. For $E>0$, let
\begin{align}
    \gamma(E):=\frac{e^{-\beta(E)H}}{\mathrm{Tr}\,(e^{-\beta(E)H})},
\end{align}
where $\beta(E)$ is the solution of the equation 
\begin{align}
    \mathrm{Tr}\left(e^{-\beta H}(H-E)\right)=0.
\end{align}
\end{defn}

The inequality of \cite{Winter_2016}, which I will refer to as Winter's inequality, then stipulates:

\begin{prop}[\cite{Winter_2016}]
Let $\mathcal{B}(\mathcal{H})$ be a type $I$ factor, let $H$ be a self-adjoint operator on $\mathcal{H}$ such that for all $\beta>0$, $e^{-\beta H}$ is trace-class. Let $E>0$, let $\rho$ and $\sigma$ be two normal states on $\mathcal{B}(\mathcal{H})$ such that 
\begin{align}\mathrm{Tr}(\rho H)\leq E,\;\;\;\; \mathrm{Tr}(\sigma H)\leq E,\end{align}where the normal states are identified with their density operators.
Let $\varepsilon>0$ and suppose that 
\begin{align}
\frac{1}{2}\|\rho-\sigma\|_1\leq\varepsilon\leq 1.
\end{align}
Then, 
\begin{align}
|S(\rho)-S(\sigma)|\leq 2\varepsilon S(\gamma(E/\varepsilon))+h(\varepsilon),
\end{align}
where 
\begin{align}
h(\varepsilon)=-\varepsilon\,\mathrm{log}\,\varepsilon-(1-\varepsilon)\mathrm{log}(1-\varepsilon).
\end{align}
\end{prop}

The idea here is to use Winter's inequality to replace the assumption on the finite-dimensional nature of the boundary Hilbert space, and to adapt the proof of \cite{Akers_2022}. One then obtains the following result:

\begin{thm}
Let $(\Lambda,(M_\lambda)_{\lambda\in\Lambda},(\mathcal{E}_\lambda)_{\lambda\in \Lambda},(E_{\lambda\mu})_{\lambda,\mu\in \Lambda,\;\lambda\geq \mu}, \omega, V_N)$ be a large $N$ quantum error-correcting code. Suppose that for some choice of $\lambda$, $M_\lambda$ is finite-dimensional of dimension $d_\lambda^2$ constant in $N$. Suppose that 
\begin{align}
\label{eq:isometryassumption}
    \|V_N^\dagger V_N|_{\mathcal{H}_{\lambda}} - Id|_{\mathcal{H}_{\lambda}}\|\leq \mu_N,
\end{align}
where $\mu_N$ decays faster than any polynomial in $N$.
Let $\ket{\Psi}\in\mathcal{H}_\lambda$ (normalized). Suppose that for every unitary operator $U^{L}_\lambda, \hat{U}^{L}_\lambda$ and $U^{R}_\lambda, \hat{U}^{R}_\lambda$ in $M^{L}_\lambda$ and $M^{R}_\lambda$, there exist unitary operators $\tilde{U}^{L}_\lambda$ and $\tilde{U}^{R}_\lambda$ (chosen in a measurable way) in $\mathcal{B}(\mathcal{H}_N^L)$ and $\mathcal{B}(\mathcal{H}_N^R)$ such that
\begin{align}
\label{eq:reconstructionassumption2}
    \|V_N U^{R}_\lambda U^{L}_\lambda\ket{\Psi}-\tilde{U}^{L}_\lambda \tilde{U}^{R}_\lambda V_N \hat{U}^{R}_\lambda \hat{U}^{L}_\lambda\ket{\Psi}\|\leq \delta_N,
\end{align}
where $\delta_N$ decays faster than any polynomial in $N$. Also suppose that there exists a self-adjoint operator $\tilde{H}_N$ on $\tilde{\mathcal{H}}^L_N:=\mathcal{H}_f\otimes\mathcal{H}^\ast_f\otimes\mathcal{H}^L_N+W^L_N\mathcal{H}^L_N$ (with $W^L_N$ defined as in \eqref{eq:wndef}) such that $e^{-\beta\tilde{H}_N}$ is trace class for all $\beta>0$, and for all polynomials $Q$ and all sequences $(\xi_N)$ decaying faster than polynomially:
\begin{align}
\label{eq:WinterCondition}
    \xi_NS(\gamma(Q(N)/\xi_N))\underset{N\rightarrow\infty}{\longrightarrow}0,
\end{align}
and that there exists a polynomial $P$ such that for the unitaries and isometries introduced before, the density matrices $\rho$ of $W^{L}_N W^{R}_N V_N \hat{U}^{R}_\lambda \hat{U}^{L}_\lambda\ket{\Psi}$ and $V_N W^{R}_\lambda W^{L}_\lambda\ket{\Psi}$ restricted to $\mathcal{B}(\tilde{\mathcal{H}}^L_N)$ satisfy
\begin{align}
    \mathrm{Tr}(\rho \tilde{H}_N)\leq P(N),
\end{align}
and that the entropies of these density matrices are polynomially bounded in $N$.
Then,
for all $\ket{\Psi}\in\mathcal{H}_\lambda$,
\begin{align}
   |S(\ket{\Psi},M^L_\lambda)+A(\mathcal{H}_{L,\lambda})-S(V_N\hat{U}^{R}_\lambda \hat{U}^{L}_\lambda\ket{\Psi},\mathcal{B}(\mathcal{H}^L_N))|\underset{N\rightarrow\infty}{\longrightarrow} 0.
\end{align}
\end{thm}

Before turning to the proof of this theorem, first note that the trace-class nature of $e^{-\beta\tilde{H}_N}$ and condition \eqref{eq:WinterCondition} deserve a bit more justification as they may look a bit abstract at first sight. However, it seems reasonable to assume them in the case of a nonabelian gauge theory at high temperature and of a code with good reconstruction properties. What should at least be true is that there exists an $H_N$ satisfying such a condition on $\mathcal{H}^L_N$: the Hamiltonian of the gauge theory. A heuristic justification goes as follows: the trace-class condition follows from the fact that the finite $N$ algebras all have type $I$, and the quantity introduced in \eqref{eq:WinterCondition} can be estimated by dimensional analysis. Specifically, in the high temperature limit of a $d$-dimensional holographic CFT on a sphere with $O(N^2)$ degrees of freedom, the temperature of the Gibbs state of energy $E$ scales (see for example \cite{https://doi.org/10.48550/arxiv.2206.14814}) like\footnote{I am grateful to David Simmons-Duffin for suggesting a reasoning based on dimensional analysis.} 
\begin{align}
T\sim \left(\frac{E}{N^2}\right)^{\frac{1}{d}},
\end{align}
and the entropy scales like 
\begin{align}
S\sim \frac{E}{T}\sim N^\frac{2}{d}E^{1-\frac{1}{d}}.
\end{align}
This means that 
\begin{align}
    \xi_NS(\gamma(Q(N)/\xi_N))\sim N^\frac{2}{d}Q(N)^{1-\frac{1}{d}}\xi_N^{\frac{1}{d}}.
\end{align}
As $\xi_N$ decays faster than any polynomial, this gives a heuristic justification for assumption \eqref{eq:WinterCondition}. Now arguably $\tilde{\mathcal{H}}_N$ is a bit larger than $\mathcal{H}_N$, so this condition on $\tilde{\mathcal{H}}_N$ can be seen as requiring an extra strength of the code. It would be interesting to see if this assumption can be improved. However, if one did not need to introduce an extra reference system of square integrable functions on the unitary group, this argument would provide a full justification of why of Fannes' inequality can be replaced by Winter's inequality in the infinite-dimensional case, in the case of a high temperature CFT.

Another remark is that one now needs to introduce the extra assumption \eqref{eq:isometryassumption} compared to the previous case. The reason is that it does not seem trivial that Winter's inequality is still valid for non-normalized states, so one needs the norm of the different states introduced to be very close to $1$. It would be interesting to find out whether there exists an analog of Winter's inequality for non-normalized states.

Let us now see how under such an assumption, the previous proof can be adapted. 

\begin{proof}
The proof of \ref{thm:main} can be adapted identically until Fannes' inequality comes into play. In the latter part of the proof, one needs to replace Fannes' inequality with Winter's inequality. 

Now denote by $V_N\hat{U}^L_\lambda\hat{U}^R_\lambda\ket{\Psi}^{norm}$ the normalized state associated to $V_N\hat{U}^L_\lambda\hat{U}^R_\lambda\ket{\Psi}$, and by $V_NW^L_\lambda W^R_\lambda\ket{\Psi}^{norm}$ the normalized state associated to $V_NW^L_\lambda W^R_\lambda\ket{\Psi}$. By the triangle inequality and assumption \eqref{eq:isometryassumption}, there exists $\delta^\prime_N$ decaying faster than polynomially such that
\begin{align}
\|W^L_NW^R_NV_N\hat{U}^L_\lambda\hat{U}^R_\lambda\ket{\Psi}^{norm}-V_NW^L_\lambda W^R_\lambda\ket{\Psi}^{norm}\|\leq\delta^\prime_N.
\end{align}
It is straightforward that the normalized states also have polynomially bounded energy (by say a polynomial $P_{norm}(N)$). Therefore applying Winter's inequality (and Lemma C.1 of \cite{Akers_2022}) yields
\begin{align}
    |S(W^L_NW^R_NV_N\hat{U}^L_\lambda\hat{U}^R_\lambda\ket{\Psi}^{norm},\mathcal{B}(\tilde{\mathcal{H}}^L_N))-S(V_NW^L_\lambda W^R_\lambda\ket{\Psi}^{norm},\mathcal{B}(\tilde{\mathcal{H}}^L_N))|\nonumber\\\leq 2\delta^\prime_NS(\gamma(P_{norm}(N)/\delta^\prime_N))+h(\delta^\prime_N).
\end{align}
Since the norm differences $\|V_N\hat{U}^L_\lambda\hat{U}^R_\lambda\ket{\Psi}^{norm}-V_N\hat{U}^L_\lambda\hat{U}^R_\lambda\ket{\Psi}\|$ and $\|V_NW^L_\lambda W^R_\lambda\ket{\Psi}^{norm}-V_NW^L_\lambda W^R_\lambda\ket{\Psi}\|$ decay faster than any polynomial in $N$, the triangle inequality and Winter's inequality allow to obtain the result: indeed 
\begin{align}
    |S(V_N\hat{U}^L_\lambda\hat{U}^R_\lambda\ket{\Psi},\mathcal{B}(\mathcal{H}^L_N))-S(V_NW^L_\lambda W^R_\lambda\ket{\Psi},\mathcal{B}(\tilde{\mathcal{H}}^L_N))|\nonumber\\\leq |S(V_N\hat{U}^L_\lambda\hat{U}^R_\lambda\ket{\Psi},\mathcal{B}(\mathcal{H}^L_N))-S(V_N\hat{U}^L_\lambda\hat{U}^R_\lambda\ket{\Psi}^{norm},\mathcal{B}(\mathcal{H}^L_N))|\nonumber\\+|S(W^L_NW^R_NV_N\hat{U}^L_\lambda\hat{U}^R_\lambda\ket{\Psi}^{norm},\mathcal{B}(\tilde{\mathcal{H}}^L_N))-S(V_NW^L_\lambda W^R_\lambda\ket{\Psi}^{norm},\mathcal{B}(\tilde{\mathcal{H}}^L_N))|&&\nonumber\\+|S(V_NW^L_\lambda W^R_\lambda\ket{\Psi}^{norm},\mathcal{B}(\tilde{\mathcal{H}}^L_N))-S(V_NW^L_\lambda W^R_\lambda\ket{\Psi},\mathcal{B}(\tilde{\mathcal{H}}^L_N))|&&.
\end{align}

The first and last term decay to zero, as $\mu_N$ decays faster than polynomially whereas all involved entropies grow at most polynomially, while the middle term decays due to Winter's inequality coupled to the assumption on the dynamics.
\end{proof}

\subsection{Type $I_\infty$ factors in the bulk}

Another possible generalization of the previous result corresponds to the case where the bulk algebra is infinite-dimensional. Of course, an infinite-dimensional code subspace cannot be encoded well in the boundary theory at finite $N$, but if one allows the code subspace dimension to grow with $N$, one can imagine a situation in which this infinite-dimensional algebra is approximated increasingly well by bigger and bigger subalgebras for each value of $N$. The mildest possible case is that in which the entropy associated to the bulk state is still $O(1)$ at large $N$, but is carried by an \textit{infinite-dimensional} factor. As shown in Appendix \ref{sec:BoundedTypeI}, boundedness of entropy for a finite-dimensional resolution of the bulk algebra implies that algebra in question must have type $I$ - since it is here supposed to be infinite-dimensional, type $I_\infty$. This is an important case as a potential choice of regulator for entanglement entropy in the bulk effective field theory could be provided by the split property \cite{https://doi.org/10.48550/arxiv.1905.00577}, which famously involves type $I_\infty$ factors. Type $I_\infty$ factors can be identified with $\mathcal{B}(\mathcal{H})$ for $\mathcal{H}$ a separable Hilbert space, which means that in this case normal states can be identified with density operators. In particular, they have a Schmidt decomposition. This allows to approximate states by finite-dimensional density matrices in a very explicit way, and to define a set of ``admissible states" for which these approximations are strong enough that the Ryu--Takayanagi formula is still valid independently of the choice of approximation.

The first step in this investigation of recovery for infinite-dimensional type $I$ factors is to introduce a general approximation procedure for a type $I_\infty$ factor in terms of a given faithful normal state and its Schmidt coefficients. 

\begin{defn}
For $M$ a type $I_\infty$ factor standardly represented on a Hilbert space $\mathcal{K}$ and $\ket{\Psi}$ a cyclic separating vector, write $M=\mathcal{B}(\mathcal{H})$ for some infinite-dimensional Hilbert space $\mathcal{H}$. The restriction of $\ket{\Psi}$ to $M$ is a trace-class density operator $\rho$ on $\mathcal{H}$, as it is a normal state on a type $I$ factor. Arrange the eigenvalues $(\lambda_1,\dots,\lambda_i,\dots)$ in decreasing order, and find a corresponding eigenbasis $(e_1,\dots,e_i,\dots)$. Now for $d\in\mathbb{N}$, decompose $(e_1,\dots,e_i,\dots)$ into $d$ families of the form $(e_{md+k})_{m\in\mathbb{N}}$, with $k$ running from $1$ to $d$. This induces a tensor product factorization of the form 
\begin{align}
    \mathcal{H}=\mathcal{H}_d\otimes\mathcal{H}^\prime_d,
\end{align}
with $\mathcal{H}_d$ finite-dimensional. For this decomposition, define 
\begin{align}
M_d:=\mathcal{B}(\mathcal{H}_d)\otimes Id,
\end{align}
and the conditional expectation onto $M_d$
\begin{align}
    E_{\Psi,d}(X\otimes Y):= \Psi(Id\otimes Y)(X\otimes Id).
\end{align}
\end{defn}

Note that 
\begin{align}\Psi\circ E_{\Psi,d}=\Psi|_{M_d}\otimes\Psi|_{M^\prime_d}.\end{align}

It is easy to show that $\Psi\circ E_{\Psi,d}$ and $\Psi$ become arbitrarily close in norm (and so do their von Neumann entropies) for a state with finite entropy as $d$ goes to infinity. However, the goal of approximating the boundary entropy of $\Psi$ with that of some $\Psi\circ E_{\Psi,d}$ cannot in general be achieved by keeping $d$ fixed as $N$ grows. Indeed, as $N$ goes to infinity Fannes' inequality (or Winter's inequality) introduces a divergent factor that needs to be cancelled by an $N$-dependent improvement of the approximation. I now introduce a class of states for which such a regulation is possible.

\begin{defn}
Let $\ket{\Psi}\in\mathcal{H}_\lambda$, cyclic separating with respect to $M_\lambda$. Let $(e_i,\lambda_i)$ be a Schmidt basis and the Schmidt coefficients associated to $\ket{\Psi}$, with Schmidt coefficients in decreasing order. Let $\eta>0$, and let $k(\eta)$ be the smallest integer such that 
\begin{align}
    \sum_{i=k(\eta)+1}^\infty\lambda_i\leq\eta.
\end{align}
Let $\mathcal{H}_{\Psi,\eta}$ be the vector space spanned by the action of $M^{k(\eta)}_\lambda$ on the vector representative $\ket{\Psi_0^\eta}$ of $\Psi\circ E_{\Psi,k(\eta)}$ in the natural cone of $\ket{\Psi}$. Denote by $M^{L,R,k(\eta)}_\lambda$ the algebra $M^{k(\eta)}_\lambda$ and its commutant represented on $\mathcal{H}_{\Psi,\eta}$.
The state $\ket{\Psi}$ is said to be \textit{admissible} for the family of maps $(V_N)$ if it has finite entropy, and there exists a sequence of thresholds $(\eta_N)_{N\in\mathbb{N}}$ such that $\sqrt{\eta_N}$ decays faster than $(\mathrm{log}\,d_N)^{-1}$, where $d_N$ is the dimension of $\mathcal{H}_N$, $\mathrm{log}\, k(\eta_N)$ grows at most polynomially, and for every unitary operator $U^{L}_\lambda, \hat{U}^{L}_\lambda$ and $U^{R}_\lambda, \hat{U}^{R}_\lambda$ in $M^{L,k(\eta_N)}_\lambda$ and $M^{R,k(\eta_N)}_\lambda$, there exist unitary operators $\tilde{U}^{L}_\lambda$ and $\tilde{U}^{R}_\lambda$ in $\mathcal{B}(\mathcal{H}_N^L)$ and $\mathcal{B}(\mathcal{H}_N^R)$ such that
\begin{align}
    \|V_N U^{R}_\lambda U^{L}_\lambda\ket{\Psi_0^{\eta_N}}-\tilde{U}^{L}_\lambda \tilde{U}^{R}_\lambda V_N \hat{U}^{R}_\lambda \hat{U}^{L}_\lambda\ket{\Psi_0^{\eta_N}}\|\leq \delta_N,
\end{align}
where $\delta_N$ decays faster than any polynomial in $N$.  
\end{defn}

Here, a few comments are in order. First, the bound given in Equation \eqref{eq:kbound} of the appendix of this paper shows that if one allows for the dimension of $\mathcal{H}_{\Psi,\eta_N}$ to scale like the exponential of a polynomial in $N$ (assuming this is the scaling of the dimension of $\mathcal{H}_N$), the restriction on the decay of $\eta_N$ is vacuous for states of bounded entropy. However, allowing the code space to be exponentially large comes with its own sets of problems, and requires new assumptions, as will soon be discussed. It would be interesting to see if the bound \eqref{eq:kbound} can be made more constraining by imposing some kind of physical condition on the state. Without trying to do this, one can however imagine an intermediate class of states, that are not invariant under any conditional expectation onto a finite-dimensional subalgebra, but for which the threshold $\eta_N$ is still saturated quickly enough (for example, polynomially in $N$). For these states, it is reasonable to keep the same assumptions as before and prove a closely related Ryu--Takayanagi formula, thanks to the following lemma.

\begin{lem}
If $\rho$ is a density matrix on $M$, for the previous factorization and for $d>0$,
\begin{align}
    \|\rho-\rho_1\otimes\rho_2\|_1\leq 4\sum_{l=d+1}^\infty\lambda_l.
\end{align}
\end{lem}

\begin{proof}
If $\lambda_1,\dots,\lambda_n,\dots$ are the Schmidt coefficients of $\rho$, for $d>0$, $1\leq k\leq d$:
\begin{align}
    (\rho_1\otimes\rho_2)_{id+k}=\left(\sum_{j=0}^\infty \lambda_{jd+k}\right)\left(\sum_{l=1}^d\lambda_{id+l}.\right)
\end{align}
Hence, for $1\leq k\leq d$, 
\begin{align}
    (\rho_1\otimes\rho_2)_{k}=\left(\sum_{j=0}^\infty \lambda_{jd+k}\right)\left(\sum_{l=1}^d\lambda_{l}\right).
\end{align}
We deduce,
\begin{align}
    \left|(\rho_1\otimes\rho_2)_{k}-\rho_k\right|=\left|\left(\sum_{j=0}^\infty \lambda_{jd+k}\right)\left(\sum_{l=1}^d\lambda_{l}\right)-\lambda_k\right|\\=\left|\left(\sum_{j=1}^\infty \lambda_{jd+k}\right)\left(\sum_{l=1}^d\lambda_{l}\right)-\lambda_k\sum_{l=d+1}^\infty\lambda_l\right|\\\leq \sum_{j=1}^\infty \lambda_{jd+k}+\lambda_k\sum_{l=d+1}^\infty\lambda_l.
\end{align}
Similarly, for $i\geq 1$:
\begin{align}
\left|(\rho_1\otimes\rho_2)_{id+k}-\rho_{id+k}\right|\leq \left(\sum_{j=0}^\infty \lambda_{jd+k}\right)\left(\sum_{l=1}^d\lambda_{id+l}\right)+\lambda_{id+k}.
\end{align}
Hence,
\begin{align}\|\rho-\rho_1\otimes\rho_2\|_1\leq\sum_{k=1}^d\sum_{j=1}^\infty \lambda_{jd+k}+\sum_{k=1}^d\lambda_k\sum_{l=d+1}^\infty\lambda_l+\sum_{k=1}^d\sum_{i=1}^\infty\left(\left(\sum_{j=0}^\infty \lambda_{jd+k}\right)\left(\sum_{l=1}^d\lambda_{id+l}\right)+\lambda_{id+k}\right)\\\leq4\sum_{l=d+1}^\infty\lambda_l.\end{align}
\end{proof}

The Ryu--Takayanagi formula for admissible states can then be formulated as follows:

\begin{thm}
Let $(\Lambda,(M_\lambda),(\mathcal{E}_{\lambda}),(E_{\lambda\mu}), \omega, V_N)$ be a renormalizable large-$N$ quantum error-correcting code. Let $\ket{\Psi}$ be a (normalized) state in $\mathcal{H}_{\lambda_0}$. Then, for all $\lambda\geq\lambda_0$ such that $\ket{\Psi}$ is admissible for $M_\lambda$, and for all $\varepsilon>0$, there exists a sequence of finite-dimensional Hilbert subspaces $\mathcal{H}_{\lambda,N}^{fin}$ of $\mathcal{H}_\lambda$ such that
\begin{align}
\label{eq:rtformula}
    |S(\ket{\Psi},M_\lambda)+A(\mathcal{H}_{\lambda,N}^{fin})-S(V_N\ket{\Psi},\mathcal{B}(\mathcal{H}_N))|\underset{N\rightarrow\infty}{\longrightarrow}0.
\end{align}
\end{thm}

\begin{proof}
If $\Psi_0^{\eta_N}$ denotes $\Psi\circ E_{\Psi,k(\eta_N)}$, by the standard continuity bound for the mapping of this state onto the natural cone, 
\begin{align}
 \|\ket{\Psi}-\ket{\Psi_0^{\eta_N}}\|\leq 2\sqrt{\sum_{l=k(\eta_N)+1}^\infty\lambda_l}.
\end{align}
Now, as $V_N$ is a contraction,
\begin{align}
    \|V_N\ket{\Psi}-V_N\ket{\Psi_0^{\eta_N}}\|\leq 2\sqrt{\sum_{l=k(\eta_N)+1}^\infty\lambda_l}.
\end{align}
By Fannes' inequality and Lemma C.1 of \cite{Akers_2022}, it follows that
\begin{align}
    |S(V_N\ket{\Psi},\mathcal{B}(\mathcal{H}_N))-S(V_N\ket{\Psi_0^{\eta_N}},\mathcal{B}(\mathcal{H}_N))||\leq 4\sqrt{\eta_N}\mathrm{log}\left(\frac{d_N}{4\sqrt{\eta_N}}\right).
\end{align}
for $N$ large enough, and this goes to zero at infinity. Then the entropy of $V_N\ket{\Psi_0^{\eta_N}}$ is controlled by exactly the same technique as in the proof of Theorem \ref{thm:main}. The only places where one should worry are the ones where the dimension $k(\eta_N)$ appears because it now grows with $N$, but since $\mathrm{log}\, k(\eta_N)$ grows at most polynomially, the obtained bounds are still strong enough. Moreover the difference in entropy of $\ket{\Psi_0^{\eta_N}}$ on $M_\lambda^{k(\eta_N)}$ and of $\ket{\Psi}$ on $M_\lambda$ goes to zero, which concludes the proof.
\end{proof}

Also note that an analogous result could have been formulated for admissible states for the case of an infinite-dimensional boundary Hilbert space, by introducing a condition allowing to apply Winter's inequality like in the previous subsection. Another remark is that as noted earlier, the code subspace renormalization scheme itself can depend on $N$. In that case, the whole proof goes through, except the last step in which it is assumed that the entropies of the regulated subalgebras get asymptotically close to the one of the type $I_\infty$ factor, because now this $I_\infty$ factor and the corresponding Schmidt coefficients are also $N$-dependent. One then needs to add this condition by hand in the definition of admissible state.

\subsection{Large codes and minimality}

It is natural to try to generalize the methods developed in this paper to codes subspaces with faster, exponential growth in $N$. This regime is also of great interest in order to study black holes and their evaporation \cite{Hayden_2019,Akers_2022,https://doi.org/10.48550/arxiv.2207.06536}. However, in this case, some assumptions made before are no longer reasonable. In particular, supposing that the map $V_N$ is very close in norm to an isometry (which was necessary in order to use Winter's inequality) no longer makes sense, because the code subspace becomes too big and $V_N$ can get a kernel. Also, unitary reconstruction in the entanglement wedge for all unitaries can no longer be true, as the code subspace can now carry enough entropy to compete with the area term and macroscopically shift the position of the Ryu--Takayanagi surface. The Ryu--Takayanagi formula then really needs to become a quantum extremal surface formula.

Fortunately, the setup of \cite{Akers_2022} allows for such a generalization, by subdividing the code subspace into a tensor product of further finite-dimensional subspaces, that are then interpreted as local degrees of freedom, and only asking for the reconstruction of product unitaries. It would be very interesting to understand how such a factorization, or a similar regulation procedure, can be systematically implemented in the large $N$ theory. Assuming such a refined structure for the renormalized code subspace, that the boundary Hilbert space is finite-dimensional, and that the logarithm of its dimension is polynomial in $N$, the results of \cite{Akers_2022} can be applied in this setup without any modification, including the result about minimality of generalized entropy in the entanglement wedge. 

A technical detail is that it is important that Fannes' inequality also holds for subnormalized states for the proof of \cite{Akers_2022} to work, as it is no longer an option to assume that $V_N$ is very close to an isometry for a very large code subspace. It would be interesting to figure out whether an analog of Winter's inequality holds for subnormalized states, and some result of this kind would become important in this regime in order to drop the finite-dimension assumption on $\mathcal{H}_N$.

At any rate, large codes seem to require even more work than the ones discussed before in order to concretely implement regularizations and a code subspace renormalization scheme that allow to formulate a family of quantum extremal surface formulae that satisfy the Susskind--Uglum prescription. It is an important question to understand how to implement such regularizations explicitly, first for small codes and then for larger codes. Understanding perturbative $G_N$ corrections and the split property in the framework of \cite{FaulknerLi} seems to be a very promising first step in this direction.

\section{Discussion}
\label{sec:Discuss}

In this paper, I introduced a new derivation of the Ryu--Takayanagi formula in the large $N$ limit of holography. This new setup makes it manifest that not all of the large $N$ von Neumann algebra can be reconstructed satisfactorily at fixed $N$ on the boundary. Instead, it must be regulated (into a type $I$ algebra if one wants to retain a finite amount of code subspace entropy in the large $N$ limit), and with this regulation comes an arbitrary choice of UV cutoff.

I then argued that an appropriate way to define the renormalization group flow between the different code subspaces (or subalgebras) is through conditional expectations that integrate out some of the degrees of freedom. From such a family of conditional expectations, one can then construct a nested family of code subspaces, each of which contributes a different amount of entropy, that decreases when high energy modes are integrated out. This nested structure reacts well with the commutant structure, so that a good notion of complementary recovery can be defined in a consistent way. It is interesting to note that the use of conditional expectations, which was previously proposed as a model of exact holography, finds a new interpretation as a renormalization group flow in the limit where the code becomes exact.

The next and last step was to prove a Ryu--Takayanagi formula: it was shown that for a state in a regulated code subspace with admissible properties with respect to the bulk-to-boundary map, the Ryu--Takayanagi formula is satisfied, with the area term associated to the code space being identified with the entropy of a Choi--Jamiolkowski state, along the lines of \cite{Akers_2022}. Then, as the bulk cutoff is changed, the variation of the area term was proven to exactly compensate that of the bulk entropy, therefore providing a full proof of the Susskind--Uglum proposal. In this new framework of quantum error correction, this proposal can be reinterpreted as a precise instance of the ER=EPR paradigm.

A few possible extensions of the result would be nice to obtain:
\begin{itemize}
    \item It is not clear how to easily single out a UV-regulated type $I$ algebra at large $N$. The most naive guess, which is to take products of single trace operators capped off at some finite number of factors, fails because these operators do not form an algebra. It would be nice to understand better how to extract type $I$ algebras carrying a finite amount of entropy from a large $N$ algebra. More generally, it would be very interesting to embed the results of this paper into the framework of asymptotically isometric codes initiated in \cite{FaulknerLi}.
    \item It also seems important to extend the reasonings provided here to code subspaces that are not made out of states that are invariant under a conditional expectation, or that only are in some approximate way. This would probably be useful to better understand how area laws work in the examples of large $N$ sectors introduced in \cite{FaulknerLi}.
    \item Perhaps more ambitiously, one could try to generalize the results presented here to the case of linear spaces of operators with no product structure. Some preliminary attempts \cite{Ghosh_2018,https://doi.org/10.48550/arxiv.2202.03357} have been made in this direction in the literature.
    \item In the proofs appearing in this paper, either Fannes' or Winter's inequalities provided bounds on the entropies of the boundary states. It would be interesting to understand the assumptions required for the application of Winter's inequality better, and to derive them in explicit examples. It would also be nice to understand to what extent Winter's inequality generalizes to the case of subnormalized states.
    \item This proof most closely mimics the case of one side of a two-sided holographic black hole. It is also necessary to tackle the extra regularizations needed for the case where one considers a subregion of a CFT, and the area of the quantum extremal surface become infinite. Perhaps the canonical purification of \cite{https://doi.org/10.48550/arxiv.1905.00577} can be used to define such regulations more precisely.
    \item One can also allow the code subspace entropy to diverge as $N$ becomes large. Maybe one can identify some $N$-dependent regularizations of the code subspace that asymptote to von Neumann algebras that have type $II$ or $III$, and it would be nice to understand this better.
    \item More generally, in the case of a large code, this paper only scratched the surface of the problem, by directly applying the setup of \cite{Akers_2022}. It is an interesting problem to understand how to subdivide a large $N$ algebra into ($N$-dependent) local degrees of freedom in the bulk, so that conditions on product unitaries, or an analog of them, can be formulated in a meaningful way. The case of large codes is especially important to describe state-dependent black hole reconstruction and black hole evaporation \cite{Hayden_2019,Penington:2019npb,Gesteau:2021jzp,https://doi.org/10.48550/arxiv.2207.06536}, which makes it a particularly interesting avenue of research.
    \item It also seems very important to see how the equivalence between large $N$ entanglement and areas contributions arises in explicit models of quantum gravity. The SYK model \cite{Chandrasekaran:2022qmq}, matrix quantum mechanics \cite{Milekhin:2020zpg}, or maybe some supersymmetric field theories seem to be cases in which one could try to implement a renormalization scheme akin to the one identified in this paper.
\end{itemize}

\section*{Acknowledgements}
I am grateful to Chris Akers, Juan Felipe Ariza Mejia, Adam Artymowicz, Charles Cao, Tom Faulkner, Alex Jahn, Matilde Marcolli, Daniel Murphy, Geoff Penington, Leonardo Santilli and David Simmons-Duffin for discussions and correspondence. I would also like to thank Chris Akers, Tom Faulkner, Matilde Marcolli and Daniel Murphy for comments on a draft of this paper.

\appendix
\section{Bounded entropy implies type $I$}
\label{sec:BoundedTypeI}
This work mainly considers the case in which the code subspace of the holographic code contains a \textit{finite} $O(1)$ amount of entropy in the large $N$ limit. Of course, this is impossible if the full algebra of observables in a region is taken into account due to UV divergences, but as was stressed above, this calculation would not be physical anyway because the code would break at a scale that is parametrically large in $N$. However, if one regulates the algebra of the EFT by simply considering one of its subalgebras that carries finite entropy, it is possible to make sense of bulk entropy in the large $N$ limit. The goal of this appendix is to characterize those algebras that can carry a finite amount of entropy. The answer turns out to be that these algebras must have type $I$. This can already be guessed at an intuitive level from the trace structure of von Neumann algebras: the only algebras that have a non-renormalized trace are type $I$ algebras.

A more precise result can be shown by adapting an argument formulated by Matsui \cite{matsui2013boundedness} in the context of the mathematical study of spin chains: if there is a cyclic separating state on the Hilbert space whose entropy for a finite-dimensional resolution of $M$ is bounded, then $M$ has type $I$. The rest of this appendix is dedicated to the proof of the following theorem:

\begin{thm}
Let $M$ be a von Neumann factor acting on a Hilbert space $\mathcal{H}$, and let $\ket{\Psi}\in\mathcal{H}$ be cyclic separating with respect to $M$. Suppose that there exists an increasing sequence of finite-dimensional unital simple subalgebras $(M_n)_{n\in\mathbb{N}}$ of $M$ such that \begin{align}\underset{n\in\mathbb{N}}{\mathrm{sup}}\,S(\ket{\Psi},M_n)<\infty,\end{align}
and $M$ is the closure of the union of the $M_n$ for the weak operator topology. Then, $M$ has type $I$.
\end{thm}

\begin{proof}
For this proof, it will be easier to recast the problem at the level of $C^\ast$-algebras and norm closures, which is the goal of this first lemma.

\begin{lem}
Let $\mathcal{A}$ be the norm closure of the union of the $M_n$. Then, $\mathcal{H}$ is isomorphic to the GNS representation of $\mathcal{A}\otimes\mathcal{A}^{op}$\footnote{The opposite algebra is isomorphic to $J\mathcal{A}J$.} in the state defined by
$\psi(X\otimes JYJ):=\bra{\Psi}XJYJ\ket{\Psi}$.
\end{lem}
\begin{proof}
The vector $\ket{\Psi}$ is cyclic separating with respect to $M$, so one can apply Tomita--Takesaki theory to obtain the characterization of the commutant 
\begin{align}
M^\prime=JMJ,
\end{align}
where $J$ is the modular conjugation associated to $\ket{\Psi}$. 

Now let $\mathcal{A}$ be the be the norm closure of the union of the $M_n$. The (nuclear) $C^\ast$-algebra $\mathcal{A}\otimes\mathcal{A}^{op}$ has a representation on $\mathcal{H}$ sending the operator $A\otimes JBJ$ to $AJBJ$. This representation coincides with the GNS representation of the state defined by the composition of $\psi$ and this representation. The Hilbert space is isomorphic to $\mathcal{H}$, and the respective images of $\mathcal{A}\otimes Id$ and $Id\otimes\mathcal{A}^{op}$ are strong operator dense in $M$ and its commutant respectively, which makes the proof of the lemma straightforward.
\end{proof}

Now, the goal is to show that the state $\psi$ is quasiequivalent to a tensor product of the form $\psi_R\otimes\psi_L$ on $\mathcal{A}\otimes\mathcal{A}^{op}$. In order to do this, I will adapt a proof due to Matsui \cite{matsui2013boundedness} in the context of spin chains to the present case.

For $n\in\mathbb{N}$, $M_n$ is a type $I$ factor acting on $\mathcal{H}$, so one can write (see \cite{Stratila_2020}, paragraph 9.15)
\begin{align}
\mathcal{H}=\mathcal{H}_n\otimes\mathcal{H}_n^c,
\end{align}
where 
\begin{align}
M_n=\mathcal{B}(\mathcal{H}_n)\otimes Id.
\end{align}
With respect to this factorization, the state $\ket{\Psi}$ admits a Schmidt decomposition
\begin{align}
    \ket{\Psi}=\sum_{j=1}^{d_n}\sqrt{\lambda_j^{(n)}}\,\xi_j^{(n)}\otimes\eta_j^{(n)},
\end{align}
with $\xi_j^{(n)}\in\mathcal{H}_n$ and $\eta_j^{(n)}\in\mathcal{H}_n^c$, $1\geq\lambda_1^{(n)}\geq\dots\geq\lambda_{d_n}^{(n)}\geq 0$, and $\sum_{j=1}^{d_n}\lambda_j^{(n)}=1$.

Let \begin{align}S:=\underset{n\in\mathbb{N}}{\mathrm{sup}}\,S(\ket{\Psi},M_n).\end{align}
\begin{lem}
\label{lem:entropybounds}
Let $1>\varepsilon>0$. If $k$ is the integer defined by 
\begin{align}
\sum_{j=k}^{d_n}\lambda_j^{(n)}\geq\varepsilon
\end{align}
and 
\begin{align}
\sum_{j=k+1}^{d_n}\lambda_j^{(n)}<\varepsilon,
\end{align}
then 
\begin{align}
\label{eq:kbound}
k\leq\mathrm{exp}\left(\frac{S}{\varepsilon}\right),
\end{align}
and
\begin{align}
\label{eq:lambda1bound}
\lambda_1^{(n)}\geq\mathrm{exp}\left(-\frac{S}{\varepsilon}\right).
\end{align}
\end{lem}
\begin{proof}
As the logarithm is an increasing function, 
\begin{align}
-\varepsilon\,\mathrm{log}\lambda_k^{(n)}\leq -\sum_{j=k}^{d_n}\lambda_j^{(n)}\mathrm{log}\lambda_k^{(n)}\leq S.
\end{align}
This already proves \eqref{eq:lambda1bound}. Then it suffices to note that
\begin{align}
k\lambda_k^{(n)}\leq\sum_{j=1}^{k}\lambda_j^{(n)}\leq 1
\end{align}
to obtain the other bound \eqref{eq:kbound}.
\end{proof}

Now, let us relabel the two tensor factors in $\mathcal{A}\otimes\mathcal{A}^{op}$ by $\mathcal{A}_R$ and $\mathcal{A}_L$. We have:

\begin{lem}
\label{lem:quasiequivalent}
Let $\psi_j^{(n)}$ be an extension of $\xi_j^{(n)}$ to $\mathcal{A}_R$, and let $\varphi_j^{(n)}$ be an extension of $\eta_j^{(n)}$ to $\mathcal{A}_L$. Consider $\psi_{R,j}$ and $\varphi_{L,j}$, two weak-$^\ast$ limits of the sequences $\psi_j^{(n)}$ and $\varphi_j^{(n)}$. Up to passing to a subsequence one can also assume convergence of the $\lambda_j^{(n)}$. If $\underset{n\rightarrow\infty}{\mathrm{lim}}\lambda_j^{(n)}$ is nonzero (which is the case for $\lambda_1$), then $\psi_{R,j}$ is quasi-equivalent to $\Psi_R$, and $\varphi_{L,j}$ is quasi-equivalent to $\Psi_L$, where $\Psi_R$ and $\Psi_L$ are the restrictions of the state $\ket{\Psi}$ to the algebras $\mathcal{A}_R$ and $\mathcal{A}_L$ respectively.
\end{lem}

\begin{proof}
It suffices to note that on any $M_n$, 
\begin{align}
\Psi_R=\sum\lambda_j^{(n)}\psi_{j}^{(n)}\geq\lambda_{j_0}^{(n)}\psi_{j_0}^{(n)}
\end{align}
for all $j_0$. By going to the limit, 
\begin{align}
    \left(\underset{n\rightarrow\infty}{\mathrm{lim}}\lambda_{j_0}^{(n)}\right)\psi_{R,j_0}\leq\Psi_R.
\end{align}
This inequality applied to $\underset{n\rightarrow\infty}{\mathrm{lim}}\lambda_1^{(n)}\neq 0$, together with the fact that the GNS representation associated to $\Psi_R$ is a factor by assumption, shows that the two representations are quasi-equivalent. The same reasoning can be applied to $\mathcal{A}_L$.
\end{proof}

We are now ready to show that $\Psi$ is quasi-equivalent to $\Psi_L\otimes\Psi_R$. Lemma \ref{lem:quasiequivalent} shows that it is enough to demonstrate that it is quasi-equivalent to $\Psi_{L,1}\otimes\Psi_{R,1}$. For this last part of the proof, I will closely follow the notations of \cite{matsui2013boundedness}.

Let $1>\varepsilon >0$, and let $K$ be defined as the largest integer smaller or equal to $\mathrm{exp}\left(\frac{S}{\varepsilon}\right)$. Define the vectors (for notational simplicity, bra-ket notation is not explicitly used for these) \begin{align}\tilde{\Omega}(n):=\sum_{j=1}^K\sqrt{\lambda_j^{(n)}}\xi_j^{(n)}\otimes\eta_j^{(n)},\end{align} and 
\begin{align}\Omega(n):=\frac{\tilde{\Omega}(n)}{\|\tilde{\Omega}(n)\|}.\end{align}
We then have that 
\begin{align}
    0<1-\|\tilde{\Omega}(n)\|^2<\varepsilon,
\end{align}
\begin{align}
    1-\|\tilde{\Omega}(n)\|<\frac{\varepsilon}{1+\|\tilde{\Omega}(n)\|}<\varepsilon,
\end{align}
and
\begin{align}
    \|\tilde{\Omega}(n)-\ket{\Psi}\|^2<\varepsilon.
\end{align}
\begin{align}
    \|\Omega(n)-\ket{\Psi}\|^2=\left(\frac{1}{\|\tilde{\Omega}(n)\|}-1\right)^2\left(\sum_{j=1}^{K}\lambda_j^{(n)}\right)+\sum_{j=K+1}^{d_n}\lambda_j^{(n)}=(\|\tilde{\Omega}(n)\|-1)^2+\sum_{j=K+1}^{d_n}\lambda_j^{(n)}\leq 2\varepsilon.
\end{align}
Now consider $\omega_\infty$, a weak-$^\ast$ accumulation point of the $\omega_n$, linear functionals associated to the $\Omega(n)$. It follows that 
\begin{align}
    \|\omega_n-\Psi\|\leq2\sqrt{2\varepsilon},
\end{align}
and going to the limit,
\begin{align}
\label{eq:closeness}
    \|\omega_\infty-\Psi\|\leq2\sqrt{2\varepsilon}.
\end{align}

Now, the Cauchy--Schwarz inequality gives on any $M_{n_0}\otimes\mathcal{A}_L$, for $n\geq n_0$,
\begin{align}
    \omega_n\leq \frac{K}{1-\varepsilon}\sum_{j=1}^K\lambda_j^{(n)}\psi_j^{(n)}\otimes\varphi_j^{(n)}.
\end{align}
Going once again to the limit (one can suppose that all the $\lambda_j^{(n)}$ converge), we get, on a dense set and hence on the full $\mathcal{A}_R\otimes\mathcal{A}_L$, 
\begin{align}
\label{eq:omegainequality}
    \omega_\infty\leq \frac{K}{1-\varepsilon}\sum_{j=1}^{K_0}\bar{\lambda}_j\psi_{j,R}\otimes\varphi_{j,L},
\end{align}
where the $\bar{\lambda}_j$ are the $K_0\leq K$ nonzero limits of the $\lambda_j^{(n)}$.
If one defines a state $\tilde{\Psi}$ such that 
\begin{align}
    C\tilde{\Psi}=\frac{K}{1-\varepsilon}\sum_{j=1}^{K_0}\bar{\lambda}_j\psi_{j,R}\otimes\varphi_{j,L},
\end{align}
for some constant $C$, then $\tilde{\Psi}$ is a linear combination of the $\psi_{j,R}\otimes\varphi_{j,L}$ with nonzero coefficients, so it is quasi-equivalent to $\Psi_R\otimes\Psi_L$, by Lemma \ref{lem:quasiequivalent}. As a consequence, it is a factor state, and $\omega_\infty$ is quasiequivalent to it by Equation \eqref{eq:omegainequality}. Equation \eqref{eq:closeness}, combined with Theorem 2.7 of \cite{10.2307/1970364}, then implies by the triangle inequality that $\Psi$ is quasiequivalent to $\Psi_R\otimes\Psi_L$. 

From this last fact, one can deduce that the von Neumann algebra $M$ has type $I$. Indeed, $\Psi$ is pure so it is type $I$. It is quasiequivalent to the tensor product of $\Psi_R$ and $\Psi_L$, which means that both must also be type $I$ (as a tensor product is type $I$ if and only if both factors are type $I$). Now $M$ corresponds to a factor representation, so every one of its subrepresentations is quasiequivalent to it. In particular, this is true for the GNS representation of $\Psi_R$, which shows that $M$ has type $I$  (see Lemma 4.3 of \cite{Naaijkens_2022}).
\end{proof}

\bibliographystyle{myJHEP}
\bibliography{references}
\end{document}